\documentclass{article}

\usepackage{setspace}
\usepackage[T1]{fontenc} % necessary for "old text encoding" apparently (like for \k{a} (ogonek used in Polish))

\usepackage{graphicx}
\usepackage{amsthm,amsmath,amssymb,enumitem,mathrsfs}
\usepackage{color}
\usepackage[noadjust]{cite}
\usepackage{multirow}
\usepackage{tabu,array,soul}
\usepackage{subcaption}
\usepackage{caption}
\usepackage[ruled,vlined,linesnumbered]{algorithm2e}

\usepackage{hyperref}
\hypersetup{
  colorlinks = true,
  linkcolor = black!30!red,
  citecolor = black!30!green
}

\usepackage{siunitx,cleveref}

\usepackage{textcomp}
\usepackage{tikz}
\usetikzlibrary{decorations.pathreplacing, calc}
\usetikzlibrary{positioning}
\usetikzlibrary {arrows.meta}
\usepackage{rotating}
\usepackage{amsbsy}
\usepackage[textsize=footnotesize,color=green!40]{todonotes}

\newtheorem{theorem}{Theorem}
\newtheorem{definition}{Definition}
\newtheorem{proposition}[theorem]{Proposition}
\newtheorem{corollary}{Corollary}

\newtheorem{lemma}[theorem]{Lemma}  

\newtheorem{claim}[theorem]{Claim}

\newtheorem{observation}[theorem]{Observation}

\title{Algorithms and complexity for geodetic sets on partial grids \thanks{This research was supported by the IFCAM project ``Applications of graph homomorphisms'' (MA/IFCAM/18/39), by the ANR project HOSIGRA (ANR-17-CE40-0022), and by an ISIRD Grant from Sponsored Research and Industrial Consultancy, IIT Kharagpur. This paper contains the full versions of parts of an extended abstract from the proceedings of ISAAC 2020~\cite{ISAAC-version}.}} %TODO Please add

\author{Dibyayan Chakraborty\thanks{School of Computing, University of Leeds, United Kingdom.} \and  Harmender Gahlawat \thanks{Ben-Gurion University of the Negev, BeerSheba, Israel.}  \and Bodhayan Roy \thanks{Indian Institute of Technology, Kharagpur. } }

\begin{document}

\maketitle

%\setstretch{1.1}
\begin{abstract}
 A set $S$ of vertices of a graph $G$ is a \emph{geodetic set} if every vertex of $G$ lies in a shortest path between some pair of vertices of $S$. The \textsc{Minimum Geodetic Set (MGS)} problem is to find a geodetic set with minimum cardinality of a given graph.  A \emph{grid embedding} of a graph is a set of points in two dimensions with integer coordinates such that each point in the set represents a vertex of the graph and, for each edge, the points corresponding to its endpoints are at Euclidean distance~$1$. A graph is a \emph{partial grid} if it has a grid embedding. In this paper, we first prove that \textsc{Minimum Geodetic Set} remains NP-hard even for subcubic partial grids of arbitrary girth. This jointly strengthens three existing hardness results: for bipartite graphs (Dourado et al., Discrete. Math, 2010), subcubic graphs (Bueno et al., Inf. Process. Lett., 2018)~\cite{bueno2018}, and planar graphs (Chakraborty et al., CALDAM, 2020). 

 { The \emph{area} of an internal face is the number of integer points lying on the boundary or interior of the face.}
  A  graph is a \emph{solid grid} if it has a grid embedding such that all interior faces have area exactly four. To complement the above hardness result, we design a linear-time algorithm for \textsc{Minimum Geodetic Set} on solid grids, improving on a $3$-approximation algorithm by Chakraborty et al. (CALDAM, 2020).

Our results hold for \textsc{Edge Geodetic Set} as well. A set $S$ of vertices of a graph $G$ is a \emph{geodetic set} if every edge of $G$ lies in a shortest path between some pair of vertices of $S$. The \textsc{Minimum Edge Geodetic Set (MEGS)} problem is to find an edge geodetic set with minimum cardinality of a given graph. As corollaries, we obtain that \textsc{MEGS} remains NP-hard on partial grids and is linear-time solvable on solid grids.
  
\end{abstract}

\section{Introduction}

A simple undirected graph $G$ has vertex set $V(G)$ and edge set $E(G)$. For two vertices $u,v\in V(G)$, let $I(u,v)$ denote the set of all vertices in $G$ that lie in some shortest path between $u$ and $v$.  For a subset $S$ of vertices of a graph $G$, let $I(S)=\bigcup_{u,v\in S} I(u,v) $. We say that $T\subseteq V(G)$ is \emph{covered} by $S$ if $T\subseteq I(S)$. A set of vertices $S$ is a \emph{geodetic set} if $V(G)$ is covered by $S$. The \emph{geodetic number}, denoted as $gn(G)$, is the minimum integer $k$ such that $G$ has a geodetic set of cardinality $k$. Given a graph $G$, the \textsc{Minimum Geodetic Set (MGS)} problem, introduced in~\cite{harary1993}, is to compute a geodetic set of $G$ with minimum cardinality. In this paper, we study the computational complexity of \textsc{Minimum Geodetic Set} in subclasses of planar graphs. \textsc{Minimum Geodetic Set} is a natural graph covering problem that falls in the class of problems dealing with the important geometric notion of \emph{convexity}: see~\cite{farber1986,bookGC} for some general discussion of graph convexities. The setting of \textsc{Minimum Geodetic Set} is quite natural, and it can be applied to facility location problems such as the optimal determination of bus routes in a public transport network etc.~\cite{caldam2020,bus}. See also~\cite{ekim2012} for further applications. The aim of this paper is to study \textsc{Minimum Geodetic Set} on \emph{partial grids}. A \emph{grid embedding} of a graph is a set of points in two dimensions with integer coordinates such that each point in the set represents a vertex of the graph and, for each edge, the points corresponding to its endpoints are at Euclidean distance~$1$. A graph is a \emph{partial grid} if it has a grid embedding. {The \emph{area} of an internal face is the number of integer points on the boundary or interior of the face.} A  graph is a \emph{solid grid} if it has a grid embedding such that all interior faces have area exactly four.

The algorithmic complexity of \textsc{Minimum Geodetic Set} has been studied intensively. \textsc{Minimum Geodetic Set} is known to be NP-hard on \emph{chordal} graphs~\cite{JCMCC96}, and \emph{(chordal) bipartite} graphs~\cite{missingBipartite,dourado2008,dourado2010}, \emph{subcubic} graphs~\cite{bueno2018}, \emph{planar} graphs~\cite{caldam2020}, \emph{co-bipartite} graphs~\cite{ekim2012}. From the perspective of parameterized complexity, \textsc{Minimum Geodetic Set} remains W[1]-hard for the parameters solution size, \emph{feedback vertex set number}, and \emph{pathwidth}, all three parameters combined~\cite{kellerhals2020}. From the viewpoint of approximability, \textsc{Minimum Geodetic Set} 
% and \textsc{Minimum Edge Geodetic Set} 
remains LOG-$\mathcal{APX}$ hard even on subcubic bipartite graphs~\cite{davot2021approximation}.  In this paper, we jointly strengthen three existing NP-hardness results: for bipartite graphs~\cite{dourado2010}, subcubic graphs~\cite{bueno2018}, and planar graphs~\cite{caldam2020}, by proving the following theorem.

\begin{theorem}\label{thm:hardness-grid}
	\textsc{Minimum Geodetic Set} is NP-hard for subcubic partial grids of girth at least $g$, for any fixed integer $g\geq 4$.%graphs in $\mathcal{PG}(3,g)$, for every $g\geq 4$. 
\end{theorem}

Note that partial grids are subclasses of many other important graph classes such as \emph{disk} graphs, \emph{rectangle intersection} graphs, etc.~\cite{clark1990,thomassen1986}. Hence, our result implies that \textsc{Minimum Geodetic Set} remains NP-hard on the aforementioned graph classes.

% In 1993, Harary et al.~\cite{harary1993}, proved that \textsc{Minimum Geodetic Set} is NP-hard. Douthat and Kong~\cite{JCMCC96} extended this to input graphs that are \emph{chordal}. Dourado et al.~\cite{dourado2008,dourado2010} proved NP-hardness for \emph{bipartite} graphs, chordal graphs (independently from~\cite{JCMCC96}) . %(\textit{i.e.} bipartite graphs with no induced cycle of order greater than~4).
% Recently, Bueno et al.~\cite{bueno2018} proved that \textsc{Minimum Geodetic Set} remains NP-hard for \emph{subcubic} graphs, and Chakraborty et al.~\cite{caldam2020} proved that \textsc{Minimum Geodetic Set} is NP-hard for \emph{planar} graphs and \emph{line} graphs. Ekim et al. proved NP-hardness for \emph{cobipartite} graphs~\cite{ekim2012}. 

On the positive side, polynomial-time algorithms to solve \textsc{Minimum Geodetic Set} are known for \emph{cographs}~\cite{dourado2010}, \emph{split} graphs~\cite{dourado2010,JCMCC96} and more generally \emph{well-partitioned chordal} graphs~\cite{wellpart}, \emph{ptolemaic} graphs~\cite{farber1986} and more generally \emph{distance-hereditary} graphs~\cite{dh}, \emph{block-cactus} graphs~\cite{ekim2012}, \emph{outerplanar} graphs~\cite{mezzini2018} and \emph{proper interval} graphs~\cite{ekim2012}. The problem is fixed-parameter-tractable for parameters \emph{tree-depth}, \emph{modular-width} and \emph{feedback edge set number}~\cite{kellerhals2020}. Chakraborty et al.~\cite{caldam2020} gave a $3$-approximation algorithm for \textsc{Minimum Geodetic Set} on solid grids. We improve this as follows.

\begin{theorem}\label{thm:solid-grid}
	There is a linear-time algorithm for \textsc{Minimum Geodetic Set} on solid grids.
\end{theorem}

We note that researchers have designed polynomial-time algorithms for various problems on solid grids~\cite{feldmann2011,keshav2012,umans1997}.

Next, we establish that our results also hold for a related problem, \textsc{Minimum Edge Geodetic Set}. A set $S$ of vertices of a graph is an \emph{edge geodetic set} if every edge lies in some shortest path between some vertices in $S$. Note that an edge geodetic set is also a geodetic set. Given a graph $G$, the problem \textsc{Minimum Edge Geodetic Set}, introduced independently in~\cite{atici2003edge} and \cite{santhakumaran2007edge}, asks to compute an edge geodetic set with minimum cardinality. The computational complexity of \textsc{Minimum Edge Geodetic Set} has been heavily studied~\cite{anand2018edge,chartrand2000geodetic, davot2021approximation,  rehmani2019edge}. In particular, we have the following results.

\begin{corollary}~\label{C:hardness}
\textsc{Minimum Edge Geodetic Set} is NP-hard for subcubic partial grids of girth at least $g$, for any fixed integer $g\geq 4$.
\end{corollary}

\begin{corollary}~\label{C:algo}
There is a linear-time algorithm for \textsc{Minimum Edge Geodetic Set} on solid grids.
\end{corollary}

%On the negative side, we prove the following result, where $\mathcal{PG}(3,g)$ denotes the class of partial grids having maximum degree at most $3$ and girth at least $g$.

\medskip\noindent \textbf{Organisation:} In Section~\ref{sec:grid-hard}, we prove Theorem~\ref{thm:hardness-grid} and Corollary~\ref{C:hardness}. In Section~\ref{sec:solid-grid}, we prove Theorem~\ref{thm:solid-grid} and Corollary~\ref{C:algo}.

\section{Hardness for partial grids}\label{sec:grid-hard}

We now prove Theorem~\ref{thm:hardness-grid}. Let $\mathcal{PG}(3,g)$ denote the class of subcubic partial grids of girth at least~$g$. Given a graph $G$, a subset $S\subseteq V(G)$ is a \emph{vertex cover} of $G$ if every edge in $E(G)$ has at least one end-vertex in $S$. The problem \textsc{Minimum Vertex Cover} is to compute a vertex cover of an input graph $G$ with minimum cardinality. To prove Theorem~\ref{thm:hardness-grid}, we reduce the NP-complete \textsc{Minimum Vertex Cover} on cubic planar graphs~\cite{mohar2001face} to \textsc{Minimum Geodetic Set} on graphs in $\mathcal{PG}(3,g)$.

\subsection{Important lemma}

We use a result of Valiant~\cite{valiant1981} which was stated by Clark et al.~\cite{clark1990} in the following form.

\begin{theorem}[\cite{valiant1981,clark1990}]\label{valiant}
    A planar graph $G$ with maximum degree~$4$ can be embedded in the plane using $O(|V(G)|)$ area in such a way that its vertices are at integer coordinates and its edges are drawn so that they are made up of line segments of the form $x = i$ or $y = j$, for integers $i$ and $j$.
\end{theorem}

In the above theorem, ``area'' refers to the number of integer points of the grid that are covered by the edges or the internal faces of the embedding of the graph. A graph $H$ is an \emph{equal subdivision} of a graph $G$ if there exists an integer $\ell$ such that $H$ can be obtained by replacing each edge of $G$ by a path of length $\ell$. Below we use  Theorem~\ref{valiant} to prove the following.

{
\begin{lemma} \label{lem:good-embedding}
    Let $G$ be a planar graph with maximum degree~$4$. Then there exists a partial grid graph $H$, which is an equal subdivision of $G$ and contains at most $O(|V(G)|^3)$ vertices.
\end{lemma}
\begin{proof}

    First, we introduce the following definition. Consider a square $S$ whose sides are of length $a$ units for some even integer $a$. {Observe that given any even integer $b$ that satisfies $a\leq b\leq (a+1)^2-1$,} it is possible to find a rectilinear path $P$ whose length is $b$ units, $P$ lies entirely inside $S$ and the left and right endpoints of $P$ are the left bottom and right bottom corners of $S$, respectively. We call $P$ a ``zigzag expansion of the bottom segment'' of $S$. See Figure~\ref{fig:zigzag} for an example. %For our purpose, we need that the expansion paths do not touch any boundary other than the bottom boundary. In that case, observe that, given $a>2$, we can find a square $S'$ in $S$ such that the sides of $S'$ are of length $a-2$ units each and the left bottom and right bottom endpoints of $S'$ are to the one units right to the left bottom point of $S$ and one unit left to the right bottom corners of $S$.  if we have a square whose sides are of length $a$ units, then given any length $a\leq b\leq ((a-1)^2-1)+2 = (a-1)^2+1$, it is possible to find a rectilinear path $P$ whose length is $b$ units, $P$ lies entirely inside $S$, $P$ does not touch any boundary other than the bottom boundary and the left and right endpoints of $P$ are the left bottom and right bottom corners of $S$, respectively.
    
    Now we describe the proof of the lemma. Using Theorem~\ref{valiant}, we can get an embedding $\mathcal{R}$ such that its vertices are at integer coordinates and its edges are drawn so that they are made up of line segments of the form $x=i$ or $y=j$. 
    Furthermore, the dimensions of the grid, as well as the lengths of the edges, are $O(|V(G)|)$. So, let $k_1n$ be the maximum length of an edge in the grid embedding obtained through Theorem~\ref{valiant}, where $n=|V(G)|$. (If $k_1$ is not an integer, consider $k_1=\lceil k_1 \rceil$.) The minimum length of an edge is at least $1$. We scale up the embedding, both, horizontally and vertically by a factor of $8nk_1$. 
    % where the value of $k_2$ is addressed later. 
    Now the longest edge has a length  $\ell=8k_1^2n^2$, and the shortest edge has a length of at least  $8k_1n$. Note that due to the scaling factor, all edges have even lengths. Furthermore, due to the scaling, any previously maximal horizontal or vertical unit segment (say $s$) of an edge has expanded to a length of at least  $8k_1n$ and has squares ($S_1,S_2$) of dimension $8k_1n$ on both sides. Consider the square $S \in \{S_1,S_2\}$ that lies completely above (resp., on the right) of a horizontal (resp., vertical) segment $s$. (The fact that $S$ lies completely above (resp., on the right) of $s$ can be assumed without loss of generality). We say that $S$ is the  \textit{expansion square} of the segment $s$. Moreover, observe that although a horizontal (resp., vertical) segment $s$ has a unique expansion square, a square $S$ can simultaneously be an expansion square of some horizontal segment $s$ as well as some vertical segment $s'$.
    
{ 
    Now, consider an expansion square $S$ of a horizontal segment $s$. Inside $S$, we define two disjoint \textit{zigzag squares} $Z_h$ and $Z_v$ in the following manner. For ease of presentation, let the coordinates of $S$ be as follows: the bottom left point is $(0,0)$, and the top right point is $(8k_1n,8k_1n)$.  Now, the zigzag square $Z_h$ has the bottom left point at $(4k_1n,0)$ and top right point at $(8k_1n-2,4k_1n-2)$, and the zigzag square $Z_v$ has the bottom left point at $(0,4k_1n)$ and top right point at $(4k_1n-2,8k_1n-2)$. See Figure~\ref{fig:Zv} for an illustration. Observe that both squares $Z_h$ and $Z_v$ have a side of length $4k_1n-2$ each. Therefore, in the zigzag square $Z_h$ (resp., $Z_v$), we can find a zigzag expansion of the bottom segment of $Z_h$ (resp., left segment of $Z_v$) of any even length $b$, where $b$ satisfies $4k_1n-2 \leq b \leq (4k_1n-1)^2-1= 16k_1^2n^2-8k_1n$ (here, $a= 4k_1n-2$). Moreover, it is easy to observe that the squares $Z_h$ and $Z_v$ are indeed disjoint.
}

    Since the maximum difference between the lengths of the longest and shortest edges is $8k_1^2n^2- 8k_1n$, observe that, a zigzag expansion of a horizontal (resp., vertical) segment $s$ in the zigzag square $Z_h$ (resp. $Z_v$) increases the length of the edge to the maximum length $\ell$. Moreover, since each segment uses its unique and disjoint zigzag square, these zigzag paths have no intersection. After  equalizing the length of all edges in the embedding $\mathcal{R}$, an equal subdivision $H$ of $G$ can be found by replacing each edge of $G$ by a path of length $\ell$. Clearly, $H$ is a partial grid. Since the length of each edge is $O(n^2)$, and there are $O(n)$ edges, the total number of vertices in $H$ is $O(|V(G)|^3)$.
\end{proof}}

% \begin{corollary}
%     Given a planar graph  of 
% \end{corollary}

\begin{figure}
    \centering
    \includegraphics[scale=1.5]{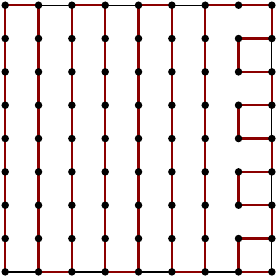}
    \caption{Here $a = 8$. We display a zigzag expansion of the bottom segment of length 80 in color red. We display the lattice points for convenience.}
    \label{fig:zigzag}
\end{figure}

\begin{figure}
        \centering
        \includegraphics[scale = 1]{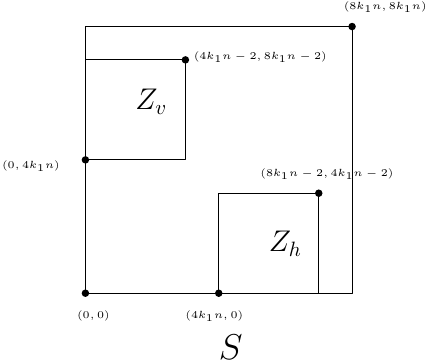}
        \caption{Illustration for zigzag squares.}
        \label{fig:Zv}
    \end{figure}

% \color{green}
% Let $\mathcal{R}$ be an embedding of a cubic planar graph $G$ as described in Theorem~\ref{valiant}. Observe that one can ensure that the distance between two vertices is at least $100$, and two parallel lines are at distance at least $100$. (Any large constant would be sufficient). We call such an embedding a \emph{good embedding} of $G$.

%\medskip

% \begin{definition}
% A set $S$ of vertices of a graph is an \emph{edge geodetic set} if every edge lies in some shortest path between some vertices in $S$.    
% \end{definition}
%  Note that an edge geodetic set is also a geodetic set. We need the following lemma.
% (see Appendix~\ref{appendix:partialgrids} for a proof).

Next, we have the following lemma.
\begin{lemma}\label{lem:subdivide}
    Let $G$ be a graph having an edge geodetic set $S$. Let $H$ be an equal full subdivision of $G$. Then $S$ is an edge geodetic set of $H$. 
\end{lemma}

\begin{proof}
Let $uv$ be an edge of $H$. Observe that $uv$ must belong to some path that was introduced to replace an edge $e$ of $G$. Let $u_e,v_e\in S$ be such that $e$ belongs to a shortest path $P$ between $u_e$ and $v_e$. Let $P'$ be the path obtained by replacing each edge of $P$ with a path having $\ell$ edges. Observe that $P'$ is a shortest path between $u_e$ and $v_e$ in $H$. Hence, $uv$ belongs to a shortest path between $u_e$ and $v_e$ in $H$.  Thus, $S$ is an edge geodetic set of $H$.
\end{proof}

{ We note that the converse of Lemma~\ref{lem:subdivide} may not be true in general. See Figure~\ref{fig:converse-not-true} for an example.}
\begin{figure}
    \centering
    \includegraphics{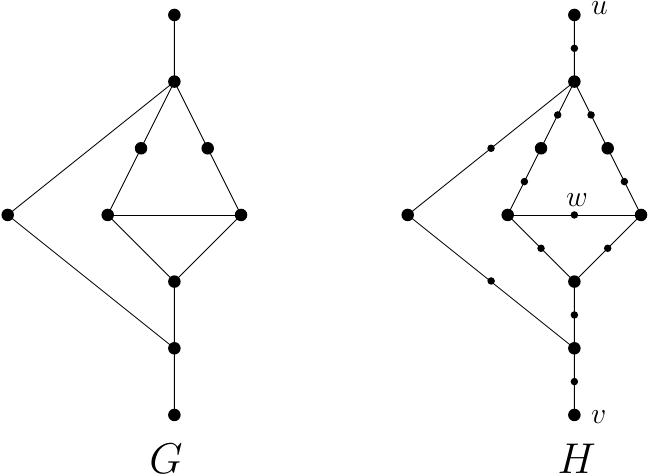}
    \caption{Here, $H$ is an equal subdivision of $G$. Observe that while the edge geodetic number of $G$ is at least $4$, $\{u,v,w\}$ is an edge geodetic set of $H$ of size $3$.}
    \label{fig:converse-not-true}
\end{figure}

\subsection{Overview of the reduction} 

Given a cubic planar graph $G$, first, we construct a planar graph $f_1(G)$ having maximum degree at most~$6$. To construct $f_1(G)$, we essentially replace the vertices of $G$ with a gadget $G_v$ (shown in Figure~\ref{fig:grid-reduction-1}).  We show that $G$ has a vertex cover of size $k$ if and only if $f_1(G)$ has an (edge) geodetic set of size $3|V(G)| + k$. Then we construct a partial grid $f_2(G)$ such that the (edge) geodetic numbers of $f_2(G)$ and $f_1(G)$ are equal. To construct  $f_2(G)$, first, we apply Lemma~\ref{lem:good-embedding} on $G$ to get a partial grid $H$ which is an equal subdivision of $G$. Let us call the vertices of $G$ that are also present in $H$ as the \textit{original vertices}. Then we replace the original vertices of $H$ with the gadget shown in Figure~\ref{fig:grid-reduction-2}, which is motivated by the construction of $G_v$ and exhibits the same properties needed to conclude the result. Finally, if the girth of the resulting graph is less than the value of $g$ in the statement of Theorem~\ref{thm:hardness-grid}, then we consider an equal subdivision of $f_2(G)$ that has girth at least $g$. {  We show using  Lemma~\ref{lem:subdivide} and some other observations that the (edge) geodetic set of the constructed graph remains unchanged, and conclude.}

% From a cubic planar graph $G$ with a given good embedding, first we construct a planar graph $f_1(G)$ having maximum degree at most~$6$ and girth~$4$. We show that $G$ has a vertex cover of size $k$ if and only if $f_1(G)$ has a geodetic set of size $3|V(G)| + k$. Then, we construct a graph $f_2(G) \in \mathcal{PG}(3,42)$ such that the geodetic numbers of $f_2(G)$ and $f_1(G)$ are the same. When $g>42$, we construct a graph $f_3(G) \in \mathcal{PG}(3,g)$ such that the geodetic numbers of $f_3(G)$ and $f_2(G)$ are the same. 

\subsection{Reduction and proofs for $f_1(G)$: The first step}

% \medskip \noindent \textbf{Construction of $f_1(G)$.} 
%with a given good embedding $\mathcal{R}$
From a cubic planar graph $G$ with a given planar embedding $\mathcal{R}$, we construct the graph $f_1(G)$ as follows. Each vertex $v$ of $G$ will be replaced by a \emph{vertex gadget} $G_v$ which is shown in Figure~\ref{fig:grid-reduction-1}. The edges outside of the vertex-gadgets will depend on $\mathcal{R}$. We assume that the edges incident with each vertex $v$ of $G$ are labeled $e_i^v$ with $0\leq i < 3$, in such a way that the numbering increases counterclockwise around $v$ with respect to the embedding (thus the edge $vw$ will have two labels: $e_i^v$ and $e_j^w$). Consider two vertices $v$ and $w$ that are adjacent in $G$, and let $e_i^v$ and $e_j^w$ be the two labels of edge $vw$ in $G$. Add the edges
 % $\{t_i^v,t_j^w\}$, $\{y_{i,i+1}^v,y_{j-1,j}^w\}$ and $\{y_{i-1,i}^v, y_{j+1,j}^w\}$
% $\{t_i^v,t_j^w\}$, $\{y_{i,i+1}^v,y_{j,j-1}^w\}$ and $\{y_{i,i-1}^v, y_{j,j+1}^w\}$
 $\{t_i^v,t_j^w\}$, $\{y_{i,i+1}^v,y_{j-1,j}^w\}$ and $\{y_{i-1,i}^v, y_{j,j+1}^w\}$
(See Figure~\ref{fig:grid-reduction-1}). All indices are taken modulo $3$. There are no other edges in $f_1(G)$. Observe that $f_1(G)$ has maximum degree at most $6$ and girth $4$.

%with vertex set $\{c^v,t^v_{0},t^v_{1},t^v_{2}\}\cup\{x_{i,j}^v,y_{i,j}^v,z_{i,j}^v~|~0\leq i<j\leq 2\}$. For simplicity we will consider that $x_{i,j}^v$ and $x_{j,i}^v$ refers to the same vertex (the same holds for $y_{i,j}^v$ and $y_{j,i}^v$, and for $z_{i,j}^v$ and $z_{j,i}^v$). There are no other vertices in $G_v$. Now we define the edge set of $G_v$ as follows. The vertex $t_i^v$ (for $0\leq i\leq 2$) is adjacent to vertices $c^v$, $y_{i,i+1}^v$, $y_{i-1,i}^v$ (indices taken modulo~$3$). Moreover, for each pair $i,j$ with $0\leq i<j\leq 2$, $x_{i,j}^v$ is adjacent to $c^v$ and $y_{i,j}^v$, and $y_{i,j}^v$ is adjacent to $z_{i,j}^v$. There is no other edge in $G_v$.

%We now describe the edges outside of the vertex-gadgets. 

\begin{figure}[t]
	\centering	
	\scalebox{1}{\begin{tikzpicture}[scale=1]
		\node[circle, draw=black!100,fill=red,thick, inner sep=0pt, minimum size=1.5mm,label=above right:{$c^v$}] (c^v) at (0,0) {};
		
		\node[circle, draw=black!100,fill=black!100,thick, inner sep=0pt, minimum size=1.5mm,label=above right:{$t_0^v$}] (t_0^v) at (0:2.5) {};

		\node[circle, draw=black!100,fill=black!100,thick, inner sep=0pt, minimum size=1.5mm,label=left:{$x_{0,1}^v$}] (x_{0,1}^v) at (60:1) {};
		\node[circle, draw=black!100,fill=black!100,thick, inner sep=0pt, minimum size=1.5mm,label=above left:{$y_{0,1}^v$}] (y_{0,1}^v) at (60:2) {};
		\node[circle, draw=black!100,fill=red,thick, inner sep=0pt, minimum size=1.5mm,label=right:{$z_{0,1}^v$}] (z_{0,1}^v) at (60:3) {};
		\node[circle, draw=black!100,fill=black!100,thick, inner sep=0pt, minimum size=1.5mm,label=left:{$t_1^v$}] (t_1^v) at (120:2.5) {};

		\node[circle, draw=black!100,fill=black!100,thick, inner sep=0pt, minimum size=1.5mm,label=above left:{$x_{1,2}^v$}] (x_{1,2}^v) at (180:1) {};
		\node[circle, draw=black!100,fill=black!100,thick, inner sep=0pt, minimum size=1.5mm,label=above left:{$y_{1,2}^v$}] (y_{1,2}^v) at (180:2) {};
		\node[circle, draw=black!100,fill=red,thick, inner sep=0pt, minimum size=1.5mm,label=below:{$z_{1,2}^v$}] (z_{1,2}^v) at (180:3) {};
		\node[circle, draw=black!100,fill=black!100,thick, inner sep=0pt, minimum size=1.5mm,label=left:{$t_2^v$}] (t_2^v) at (240:2.5) {};

		\node[circle, draw=black!100,fill=black!100,thick, inner sep=0pt, minimum size=1.5mm,label=left:{$x_{2,0}^v$}] (x_{2,0}^v) at (300:1) {};
		\node[circle, draw=black!100,fill=black!100,thick, inner sep=0pt, minimum size=1.5mm,label=below left:{$y_{2,0}^v$}] (y_{2,0}^v) at (300:2) {};
		\node[circle, draw=black!100,fill=red,thick, inner sep=0pt, minimum size=1.5mm,label=right:{$z_{2,0}^v$}] (z_{2,0}^v) at (300:3) {};

		\draw[very thick] (t_0^v)--(c^v)--(x_{0,1}^v)--(y_{0,1}^v)--(z_{0,1}^v);
		\draw[very thick] (t_1^v)--(c^v)--(x_{1,2}^v)--(y_{1,2}^v)--(z_{1,2}^v);
		\draw[very thick] (t_2^v)--(c^v)--(x_{2,0}^v)--(y_{2,0}^v)--(z_{2,0}^v);
		\draw[very thick] (t_0^v)--(y_{0,1}^v)--(t_1^v)--(y_{1,2}^v)--(t_2^v)--(y_{2,0}^v)--(t_0^v);
		
		\filldraw[fill=black!10!white,opacity=0.3,draw=black] (-3.4,-3) rectangle (3.2,3);
        
        \node[left] at (-3.5,3) {\Large $G_v$};
        
		\node[xshift=8cm,circle, draw=black!100,fill=black!100,thick, inner sep=0pt, minimum size=1.5mm,label=above left:{$c^w$}] (c^w) at (0,0) {};
		\node[xshift=8cm,circle, draw=black!100,fill=black!100,thick, inner sep=0pt, minimum size=1.5mm,label=above left:{$t_0^w$}] (t_0^w) at (180:2.5) {};
		
		\node[xshift=8cm,circle, draw=black!100,fill=black!100,thick, inner sep=0pt, minimum size=1.5mm,label=right:{$x_{0,1}^w$}] (x_{0,1}^w) at (240:1) {};
		\node[xshift=8cm,circle, draw=black!100,fill=black!100,thick, inner sep=0pt, minimum size=1.5mm,label=below right:{$y_{0,1}^w$}] (y_{0,1}^w) at (240:2) {};
		\node[xshift=8cm,circle, draw=black!100,fill=red,thick, inner sep=0pt, minimum size=1.5mm,label=left:{$z_{0,1}^w$}] (z_{0,1}^w) at (240:3) {};
		\node[xshift=8cm,circle, draw=black!100,fill=black!100,thick, inner sep=0pt, minimum size=1.5mm,label=right:{$t_1^w$}] (t_1^w) at (300:2.5) {};
		
		\node[xshift=8cm,circle, draw=black!100,fill=black!100,thick, inner sep=0pt, minimum size=1.5mm,label=above right:{$x_{1,2}^w$}] (x_{1,2}^w) at (0:1) {};
		\node[xshift=8cm,circle, draw=black!100,fill=black!100,thick, inner sep=0pt, minimum size=1.5mm,label=above right:{$y_{1,2}^w$}] (y_{1,2}^w) at (0:2) {};
		\node[xshift=8cm,circle, draw=black!100,fill=red,thick, inner sep=0pt, minimum size=1.5mm,label=below:{$z_{1,2}^w$}] (z_{1,2}^w) at (0:3) {};
		\node[xshift=8cm,circle, draw=black!100,fill=black!100,thick, inner sep=0pt, minimum size=1.5mm,label=right:{$t_2^w$}] (t_2^w) at (60:2.5) {};
		
		\node[xshift=8cm,circle, draw=black!100,fill=black!100,thick, inner sep=0pt, minimum size=1.5mm,label=right:{$x_{2,0}^w$}] (x_{2,0}^w) at (120:1) {};
		\node[xshift=8cm,circle, draw=black!100,fill=black!100,thick, inner sep=0pt, minimum size=1.5mm,label=above right:{$y_{2,0}^w$}] (y_{2,0}^w) at (120:2) {};
		\node[xshift=8cm,circle, draw=black!100,fill=red,thick, inner sep=0pt, minimum size=1.5mm,label=left:{$z_{2,0}^w$}] (z_{2,0}^w) at (120:3) {};
		
		\draw[very thick] (t_0^w)--(c^w)--(y_{0,1}^w)--(z_{0,1}^w);
		\draw[very thick] (t_1^w)--(c^w)--(y_{1,2}^w)--(z_{1,2}^w);
		\draw[very thick] (t_2^w)--(c^w)--(y_{2,0}^w)--(z_{2,0}^w);
		\draw[very thick] (t_0^w)--(y_{0,1}^w)--(t_1^w)--(y_{1,2}^w)--(t_2^w)--(y_{2,0}^w)--(t_0^w);

             \draw[very thick,blue] (c^v) -- (t_0^v); \draw[very thick,blue, densely dashed ] (t_0^v) -- (t_0^w); \draw[very thick,blue] (t_0^w) --(c^w)-- (x_{1,2}^w) --(y_{1,2}^w)--(z_{1,2}^w);
		
		\filldraw[fill=black!10!white,opacity=0.3,draw=black] (4.6,-3) rectangle (11.3,3);
	    
        \node[right] at (11.7,3) {\Large $G_w$};

		\draw[densely dashed] (t_0^v)--(t_0^w) (y_{0,1}^v)--(y_{2,0}^w) (y_{2,0}^v)--(y_{0,1}^w);

		\end{tikzpicture}}
	\caption{Illustration of vertex gadgets in the reduction in the proof of Theorem~\ref{thm:hardness-grid}. The dashed lines indicate edges between two vertex gadgets. The red vertices illustrate the vertices of a solution. The blue path illustrates an example of a shortest path between $c^v$ and $z^w_{1,2}$. }\label{fig:grid-reduction-1}
	
\end{figure}

\begin{lemma}\label{lem:reduction-1}
	The graph $G$ has a vertex cover $D$ of size $k$ if and only if $f_1(G)$ has a geodetic set of size $3|V(G)|+k$.
\end{lemma}

\begin{proof}
We construct a geodetic set $S$ of $f_1(G)$ of size $3|V(G)|+k$ as follows. For each vertex $v$ in $G$, we add the three vertices $z_{0,1}^v, z_{1,2}^v, z_{2,0}^v$ of $G_v$ to $S$. If $v$ is in $D$, then we also add the vertex $c^v$ to $S$. Let $Q=\{(0,1),(1,2),(2,0)\}$.
	
    Let us show that $S$ is indeed a geodetic set. First, we observe that in any vertex gadget $G_v$ that is part of $f_1(G)$, the unique shortest path between two distinct vertices $z_{i,j}^v$, $z_{i',j'}^v$ has length~$4$ and goes through vertices $y_{i,j}^v$, $t_{k}^v$ and $y_{i',j'}^v$ (where $\{k\}=\{i,j\}\cap\{i',j'\}$). Thus, it only remains to show that the vertices $\{c^v, x_{i,j}^v\}$ ($(i,j)\in Q$) belong to some shortest path of vertices of $S$. Assume that $v$ is a vertex of $G$ in $D$. The shortest paths between $c^v$ and $z_{i,j}^v$ have length~$3$ and one of them goes through vertex $x_{i,j}^v$. Thus, all vertices of $G_v$ belong to some shortest path between vertices of $S$. Now, consider a vertex $w\notin D$ of $G$. Since $G$ is a cubic planar graph, all three neighbours of $w$, say, $w_1,w_2,w_3$ must lie in $D$. Let $C=\{c^{w_1},c^{w_2},c^{w_3}\}$ and $Z=\{z^w_{0,1},z^w_{1,2},z^w_{2,0}\}$. Observe that all vertices of $G_w$ lie in the set $I(C \cup Z)$. Therefore, $S$ is a geodetic set.

    For the converse, assume we have a geodetic set $S'$ of $f_1(G)$ of size $3|V(G)|+k$. We will show that $G$ has a vertex cover of size~$k$. First of all, observe that all the $3|V(G)|$ vertices of type $z_{i,j}^v$ are necessarily in $S'$, since they have degree $1$. As observed earlier, the shortest paths between those vertices already go through all vertices of type $t_i^v$ and $y_{i,j}^v$. However, no other vertex lies on a shortest path between two such vertices: these shortest paths always go through the boundary $6$-cycle of the vertex-gadgets. Let $S'_0$ be the set of the remaining $k$ vertices of $S'$. These vertices are there to cover the vertices of types $c^v$ and $x_{i,j}^v$. We construct a subset $D'$ of $V(G)$ as follows: $D'$ contains those vertices $v$ of $G$ whose vertex-gadget $G_v$ contains a vertex of $S'_0$. We claim that $D'$ is a vertex cover of $G$. Suppose by contradiction that there is an edge $vw$ of $G$ such that neither $G_v$ nor $G_w$ contains any vertex of $S'_0$. Without loss of generality assume that $e_0^v$ and $e_0^w$ are the two labels of the edge $vw$.  {  We prove the following claim:

    \begin{claim}
        Let $P$ be the vertices of degree one of $f_1(G)$, $A_1=\{ z_{1,2}^v, z_{2,0}^v, z_{0,1}^v\}$ and $A_2=\{x_{1,2}^v, x_{0,1}^v, x_{2,0}^v\}$. Then $I(P) \cap A_2 = \emptyset$. Moreover, for any vertex $u'$ of $G$ that is not a neighbour of $v$ in $G$, and any vertex $a'\in G_{u'}$, $I(a' \cup A_1) \cap A_2 = \emptyset$.  
        \end{claim}
    \begin{proof}
        Clearly no vertex of $A_2$ lies in any shortest path between vertices of $A_1$. Now let $w'$ be a neighbour of $v$ in $G$ and consider the set $A_3=\{ z_{1,2}^{w'}, z_{2,0}^{w'}, z_{0,1}^{w'} \}$. For any vertex $b \in A_1$ and a vertex $b'\in A_3$, observe that all shortest paths between $b$ and $b'$ always pass through the vertices of the boundary $6$-cycle of $G_v$. Therefore such shortest paths will not contain any vertex of $A_2$ and hence $I(A_1 \cup A_3) \cap A_2 = \emptyset$. Now consider $u'$  which is not a neighbour of $v$ in $G$, and any vertex $a'\in G_{u'}$. Then, any shortest path between $a'$ and a vertex $a''\in A_1$, will contain a vertex of the form $y^v_{i,j}$ and therefore, will not contain any vertex of $A_2$. This concludes the proof.
        % Now consider inductively the statement of the claim follows. 
    \end{proof}

    }
     Let $a$ and $b$ be the neighbours of $v$ different from $w$ in $G$. Then, $I(P \cup V(G_a) \cup V(G_b) \cup V (G_c) \cup V (G_d))$ where $c$ and $d$ are neighbours of $w$ in $G$ does not contain $x_{1,2}^v$. Due to the above claim, for any vertex $u'$ that is not a neighbour of $v$ in $G$, and any vertex $a'\in G_{u'}$, the shortest path between $a'$ and any vertex of $\{z_{1,2}^v, z_{2,0}^v, z_{0,1}^v\}$ does not contain any vertex of the form $x^v_{i,j}$. Hence, $x_{1,2}^v$ is not covered by any pair of vertices in $S'$, a contradiction. Therefore, $D'$ is a vertex cover of $G$.
\end{proof}

%	To prove Lemma~\ref{lem:reduction-1}, we construct a geodetic set $S$ of $f_1(G)$ of size $3|V(G)|+k$. For each vertex $v$ in $G$, we add the three vertices $z_{i,j}^v$ ($0\leq i<j\leq 2$) of $G_v$ to $S$. If $v$ is in $D$, we also add vertex $c^v$ to $S$. The proof of validity of this set and the proof of the converse direction of the claim, can be found in Appendix~\ref{appendix:partialgrids-2}.

% Let $D$ be a vertex cover of $G$. We construct an edge geodetic set $S$ of $f_1(G)$ of size $3|V(G)|+k$ as follows. For each vertex $v$ in $G$, we add the three vertices $z_{i,j}^v$ ($0\leq i<j\leq 2$) of $G_v$ to $S$. If $v$ is in $D$, we also add vertex $c^v$ to $S$. Using arguments similar to that of Lemma~\ref{lem:reduction-1} we have the following lemma.
% %
% \begin{lemma}\label{lem:(edge)-geodetic-set}
% 	The set $S$ is both a geodetic set of minimum cardinality and an edge geodetic set of minimum cardinality of $f_1(G)$.
% \end{lemma}

\begin{figure}[t]
	\centering \includegraphics[scale=1.2,page=1]{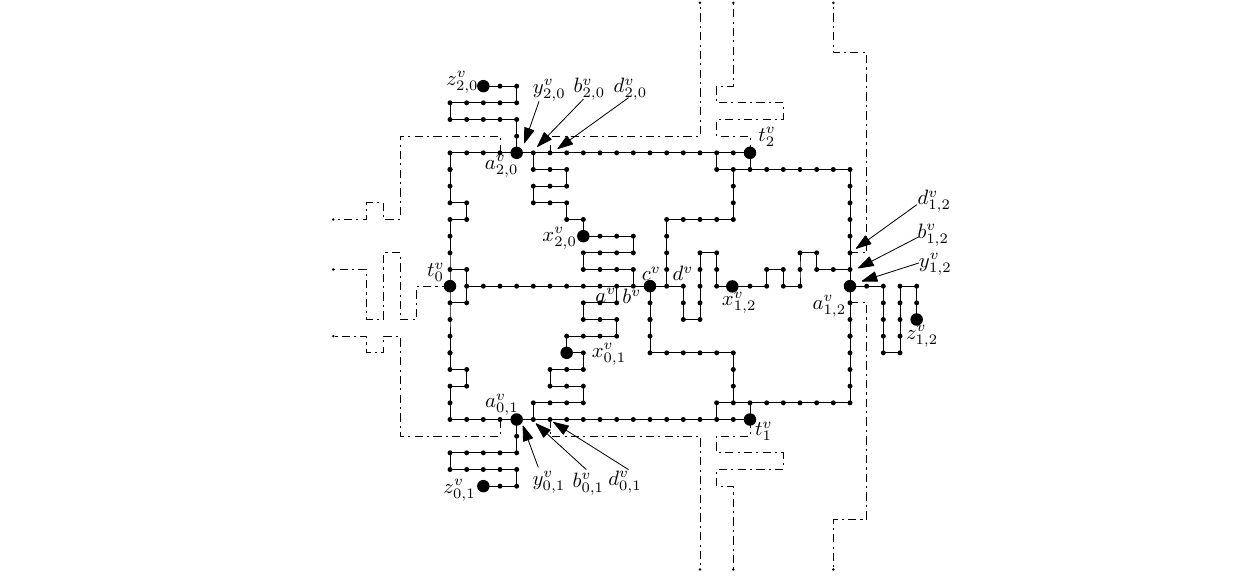} 
	\caption{Construction of $f_2(G)$.} \label{fig:grid-reduction-2}
\end{figure}

\subsection{{Reduction and proof for $f_2(G)$: The final step }}

% \medskip \noindent \textbf{Construction of $f_2(G)$.} 
In this section, we shall construct a partial grid $f_2(G)$ from a cubic planar graph and prove Theorem~\ref{thm:hardness-grid}. We shall use the gadget shown in Figure~\ref{fig:grid-reduction-2}. As before, for each vertex $v$ of $G$ we shall consider a new copy of the graph shown in Figure~\ref{fig:grid-reduction-2} and call it $G'_v$. Roughly speaking, $G'_v$ is designed such that it is a subcubic partial grid graph, and exhibits properties similar to that of $G_v$. First, we consider the following observations concerning our gadget $G'_v$ that will be useful for us.

% An edge $uw$ of $f_1(G)$ is an \emph{internal} edge if both $u$ and $w$ belongs to $G_v$ for some $v\in V(G)$. The other edges of $f_1(G)$ are \emph{external} edges. We construct $f_2(G)$ in three steps described below.

% \begin{enumerate}
% 	\item Replace each vertex of type $t^v_i$ with a new edge $T^v_i=(t^v_i t'^v_i)$. Replace each vertex of type $y_{i,j}^v$ with a path $Y_{i,j}^v = a_{i,j}^v~y_{i,j}^v~b_{i,j}^v~d_{i,j}^v$. Replace each vertex of type $c^v$ with a path $C^v = a^v~b^v~c^v~d^v$. (See Figure~\ref{fig:grid-reduction-2}). 
	
% 	% \item Replace each internal edge between vertices having labels $p,q$ with a new path such that the shortest path in the new graph between the vertices with label $p,q$ is of length $14$. 
	
% 	% \item For an edge $uv\in E(G)$, let $E_{uv}$ denote the set of three external edges between $G_u$ and $G_v$ in $f_1(G)$. Recall that $\mathcal{R}$ is a good embedding of $G$. Let $l_{uv}$ denote the length of the edge $uv$ in $\mathcal{R}$. Replace all three external edges $p_iq_i\in E_{uv}$ ($1\leq i\leq 3$) with three new paths $P_i$ ($1\leq i\leq 3$) such that lengths of all three paths are equal and in $O(l_{uv})$. 
% \end{enumerate}

\begin{observation}\label{obs:hard:1}
    Consider the graph concerning the gadget $G'_v$ and consider two vertices $p$ and $q$. Then all of the following hold (and are evident from the construction):
    \begin{enumerate}
        \item Let $p = c^v$. Then, if  $q\in \{y^v_{0,1}, y^v_{2,0}, y^v_{1,2}\}$, then $d(p ,q) = 26$. And, if $q\in \{ x^v_{0,1}, x^v_{2,0} , x^v_{1,2}\}$, then $d(p,q) = 13$. And, if $q\in \{t^v_0, t^v_1,t^v_2\}$, then $d(p,q) = 14$. Furthermore, each $x^v_{i,j} \in I(p, y^v_{i,j})$. 
        \item If $p,q\in \{y^v_{0,1}, y^v_{2,0}, y^v_{1,2}\}$, then $d(p,q) = 28$. Similarly, if $p,q\in \{t^v_{0}, t^v_{1}, t^v_{2}\}$, then $d(p,q) = 28$.
        \item Let $p=t^v_0$. Then, if $q\in \{y^v_{0,1}, y^v_{2,0}\}$, then $d(p,q) = 14$. If $q=y^v_{1,2}$, then $d(p,q) = 40$.
        
        %\item If $(p,q)\in \{\}$, then $d(a,b) = $. 
    \end{enumerate}
\end{observation}

% \todo[inline]{Could you please use correct notations ? For example, Item 1 of Observation 9. Could you please make Observation 8 a bit more concise ? An observation should not take half a page while each line contains only few words. I suggest for each distinct distance value, you can create an item.}

Next, we have the following observation that is crucial for our proof. %Observation~\ref{obs:hard:1} implies the following.

\begin{observation}\label{obs:hard-3}
    Consider the graph $G'_v$. Then, all the following hold:
    \begin{enumerate}
    \item Let $S_1 = \{z^v_{0,1}, z^v_{2,0},z^v_{1,2}, c^v\}$. Then, $S_1$ covers every edge of $G'_v$. Moreover, $S_1 \setminus \{c^v\}$ covers every vertex on the boundary face of $G'_v$.

    \item Let $S_2=  \{z^v_{0,1}, z^v_{2,0},z^v_{1,2}, t^v_0, t^v_1, t^v_2\}$. Then $S_2$ covers every edge of $G'_v$. 

    %\item Any set of vertices $S$ such that each vertex in $S$ is on the boundary of the gadget $G'_v$ and $S_2 \not\subseteq S$, then $S$ does not cover at least one vertex in $G'_v$.

    \item Consider the set $S_3 = \{y^v_{0,1}, y^v_{2,0}, y^v_{1,2}, t^v_0, t^v_1, t^v_2, a^v_{0,1}, d^v_{0,1}, a^v_{2,0}, d^v_{2,0}, a^v_{1,2}, d^v_{1,2}\}$. Then, $x^v_{0,1} \notin I(S_3\setminus \{t^v_2\})$. Similarly, $x^v_{2,0} \notin I(S_3\setminus \{t^v_1\})$ and $x^v_{1,2} \notin I(S_3\setminus \{t^v_0\})$.
    \end{enumerate}
\end{observation}

\begin{proof}
The proofs of items (1) and (2) follow directly from the construction and can be verified easily using Observation~\ref{obs:hard:1}. To prove (3), we show that only that $x^v_{0,1} \notin I(S_3\setminus \{t^v_2\})$ since the proof of $x^v_{2,0} \notin I(S_3\setminus \{t^v_1\})$ and $x^v_{1,2} \notin I(S_3\setminus \{t^v_0\})$ uses identical arguments. Targeting a contradiction, assume that there are two vertices $p,q \in S_3\setminus \{t^v_2\}$ such that there is a shortest path between $p$ and $q$, say $P$, containing the vertex $x^v_{0,1}$. Since all vertices in $S_3\setminus \{t^v_2\}$ are on the boundary of $G'_v$, note that to cover the vertex $x^v_{0,1}$, $P$ must contain the $a^v,b^v_{0,1}$-path containing $x^v_{0,1}$ as a subpath, which has length 23. 

%Without loss of generality, let us assume that $d(p,a^v)+d(q,b^v_{0,1}) \leq d(q,a^v)+d(p,b^v_{0,1})$.

 Thus, the length of $P \geq \min \{ d(p,a^v) + d(q,b^v_{0,1}), d(q,a^v)+d(p,b^v_{0,1}) \}+ 23$. We will verify that this contradicts that $P$ is an isometric path for every choice of $p,q\in S_3\setminus \{t^v_2\}$. First, it can be easily verified that $d(p,q) \leq 42$. Now, consider the following choices of $p,q$:
\begin{enumerate}
    \item One of $p,q$ (without loss of generality, say $p$) is from $\{ a^v_{2,0}, y^v_{2,0}, d^v_{2,0}, a^v_{1,2}, y^v_{1,2}, d^v_{1,2}\}$: In this case, note that $\min \{d(p,a^v),$ $ d(p,b^v_{0,1})\} \geq 26$. Therefore, $|P| \geq 26+23 = 49$, which contradicts that $P$ is an isometric path (since $d(p,q) \leq 42$).

    \item One of $p,q$ is $t^v_0$: Without loss of generality, let us assume that $p=t^v_0$ and $q \in \{a^v_{0,1} ,y^v_{0,1}, d^v_{0,1}, t^v_1\}$. Note that $d(p,q) \leq 28$. Moreover, $\min\{ d(p, a^v), d(p,b^v_{0,1})\} \geq 12$, and hence $|P| \geq 35$. This contradicts that $P$ is an isometric path (since $d(p,q) \leq 28$).

    \item $p,q\in \{a^v_{0,1} ,y^v_{0,1}, d^v_{0,1}, t^v_1\}$: Note that in this case, $d(p,q) \leq 15$, which contradicts that $P$ is an isometric path (since $|P| \geq 23$). 
\end{enumerate}
This completes our proof.
\end{proof}

Finally, to complete the reduction of $f_2(G)$, first, we apply Lemma~\ref{lem:good-embedding} on $G$ to get a partial grid $H$ which is an equal subdivision of $G$.  
% Now, we scale up the grid $H'$ again by a constant so that each original vertex $v$ can be replaced  with the gadget $G'_v$ to get a partial grid in $H$. 
Let us call the vertices of $G$ in $H$ as the \textit{original vertices}. Note that if $uw\in E(G)$, then there exists exactly one path in $H$ with endpoints $u$ and $w$ that does not contain any other original vertices. We call these paths as \textit{external paths}. Moreover, note that each external path has the same length; let that length be $\ell$.  Then, we replace each original vertex $v$ of $H$  with $G'_v$ to get $f_2(G)$, and for an edge $uw$ of $G$, we connect $G'_u$ and $G'_w$ with three paths of equal length $\ell$ in $f_2(G)$, as we did in the case of $f_1(G)$ (where we added three edges). Note that, at this step, a further equal subdivision of $H$ (by a constant) may be necessary to ``accommodate'' $G'_v$. Furthermore, note that, due to scaling, for any $uw\in E(G)$, the three ``external paths'' between the gadgets $G'_u$ and $G'_w$ can be navigated following the ``path pattern'' of the external path between $u$ and $w$ in $H$ (since any two external paths are ``sufficiently'' far apart from each other in $H$ other than at the endpoints).  
Now, for a gadget $G_v$ in $f_1(G)$, if some external edges were having $y^v_{2,0},t^v_0, y^v_{0,1}$ as endpoints, then the corresponding external paths will have  $a^v_{2,0},t^v_0, a^v_{0,1}$, respectively as endpoints in $G'_v$ in $f_2(G)$ (see Figure~\ref{fig:grid-reduction-2}).
Similarly, if the endpoints were $y^v_{2,0},t^v_2, y^v_{1,2}$ in $f_1(G)$, then they will be $d^v_{2,0},t^v_2, d^v_{1,2}$, respectively in $f_2(G)$; and if the endpoints were $y^v_{1,2},t^v_1, y^v_{0,1}$ in $f_1(G)$, then they will be $a^v_{1,2},t^v_1, d^v_{0,1}$, respectively in $f_2(G)$. 
This completes our construction.

\medskip \noindent \textbf{Completion of the proof of Theorem~\ref{thm:hardness-grid}.} Clearly, $G'_v$ is a partial grid for each $v\in G$ (Figure~\ref{fig:grid-reduction-2}). Moreover, it is not difficult to verify that $f_2(G)$ is a partial grid that has maximum degree $3$. Let $D$ be a vertex cover of $G$ with cardinality $k$. We construct a set $S$ of cardinality $3|V(G)|+k$ as follows. For each vertex $v$ in $G$, we add the three vertices with labels $z_{i,j}^v$ ($0\leq i<j\leq 2$) to $S$. If $v$ is in $D$, we also add vertex $c^v$ to $S$. From the construction of $f_2(G)$, Observation~\ref{obs:hard-3}, and using arguments similar to the ones used in the proof of the Lemma~\ref{lem:reduction-1}, $G$ has a vertex cover of size $k$ if and only if $f_2(G)$ has an (edge) geodetic set of size $3|V(G)|+k$. %Moreover, we can prove the following. 

Now, if the value of $g$ in the statement of Theorem~\ref{thm:hardness-grid} is less than the girth of $f_2(G)$, then from the previous discussions, we have that \textsc{Minimum Geodetic Set} is NP-hard for graphs in $\mathcal{PG}(3,g)$. Otherwise, we replace every edge of $f_2(G)$ with a path of length $g$. Call this modified graph $f_3(G)$, and observe that $f_3(G)\in \mathcal{PG}(3,g)$. %By Lemma~\ref{lem:reduction-2}, $S$ is both a geodetic set of minimum cardinality and an edge geodetic set of minimum cardinality in $f_2(G)$. 
Due to Lemma~\ref{lem:subdivide}, $S$ is a geodetic set of $f_3(G)$ of cardinality $3|V(G)|+k$. 

{

Now we shall argue the converse. First we have the following observation whose proof would follow a similar line of argumentation to that of Observation~\ref{obs:hard-3}.

\begin{observation}\label{obs:hard-subdivide}
    Consider an equal subdivision of the graph $G'_v$. \sloppy Consider the set $W = \{y^v_{0,1}, y^v_{2,0}, y^v_{1,2}, t^v_0, t^v_1, t^v_2, a^v_{0,1}, d^v_{0,1}, a^v_{2,0}, d^v_{2,0}, a^v_{1,2}, d^v_{1,2}\}$. Then, $x^v_{0,1} \notin I(W\setminus \{t^v_2\})$. Similarly, $x^v_{2,0} \notin I(W\setminus \{t^v_1\})$ and $x^v_{1,2} \notin I(W\setminus \{t^v_0\})$.
\end{observation}

Let $S'$ be a geodetic set of size $3|V(G)|+k$ of $f_3(G)$.  Let $S'_0\subset S'$ be the set of vertices whose degree is greater than one in $f_3(G)$. Now using the above observation and arguments similar to that of Lemma~\ref{lem:reduction-1} we can show that $G$ has a vertex cover of size $k$. This completes the proof.} Our proof also implies Corollary~\ref{C:hardness}.

\section{A linear-time algorithm for solid grids}\label{sec:solid-grid}

\newcommand{\interior}[1]{I\left(#1\right)}

\newcommand{\topleft}[1]{tl(#1)}
\newcommand{\topright}[1]{tr(#1)}
\newcommand{\bottomleft}[1]{bl(#1)}
\newcommand{\bottomright}[1]{br(#1)}

We here give a linear-time algorithm for \textsc{Minimum Geodetic Set} on solid grids and prove Theorem~\ref{thm:solid-grid}. In Section~\ref{subsec:solid-prelim}, we define some preliminary notations and state some observations. In Section~\ref{subsec:alg-solid}, we state Algorithm~\ref{alg:solid-grid} to solve \textsc{Minimum Geodetic Set} on solid grids. Then, in Section~\ref{subsec:feasible}, we prove that the set returned by our algorithm is indeed a geodetic set (Lemma~\ref{lem:feasibility}). Afterwards, in Section~\ref{sec:opt}, we prove a lower bound on the geodetic number of a solid grid (Lemma~\ref{lem:cardinlaity-2}). Finally, we prove Theorem~\ref{thm:solid-grid} in Section~\ref{subsec:solid-conclude}. The reader may take the help of Figure~\ref{fig:solid-grid-navigate} to navigate through the proof of Theorem~\ref{thm:solid-grid}. 

\begin{figure}
    \centering
    \scalebox{0.65}{
\begin{tikzpicture}
    \node (thm1) at (-0.5,0) [draw,thick,minimum width=2cm,minimum height=1cm] {Theorem~\ref{thm:solid-grid}};

    \node (obs7) at (-0.5,-2) [draw,thick,minimum width=2cm,minimum height=1cm] {Observation~\ref{obs:caardinality-solid}};

    \node (lem17) at (-8,-2) [draw,thick,minimum width=2cm,minimum height=1cm] {Lemma~\ref{lem:feasibility}};

    \node (lem22) at (7,-2) [draw,thick,minimum width=2cm,minimum height=1cm] {Lemma~\ref{lem:cardinlaity-2}};

     % \node (lem23) at (1,-2) [draw,thick,minimum width=2cm,minimum height=1cm] {Lemma~\ref{lem:time-complexity}};

    \node (lem9) at (-13,-4) [draw,thick,minimum width=2cm,minimum height=1cm] {\begin{tabular}{c} Lemma~\ref{lem:intersect-opposite} \end{tabular} };

    \node (lem14) at (-10,-4) [draw,thick,minimum width=2cm,minimum height=1cm] {\begin{tabular}{c}
    Lemma~\ref{lem:exist-cov-1}
    \end{tabular} };

    \node (lem15) at (-7,-4) [draw,thick,minimum width=2cm,minimum height=1cm] {\begin{tabular}{c}
      Lemma~\ref{lem:cov-1} \\ Lemma~\ref{lem:cov-2}
    \end{tabular} };

    \node (obs12) at (-4,-4) [draw,thick,minimum width=2cm,minimum height=1cm] {\begin{tabular}{c}
    Observation~\ref{obs:block}
    \end{tabular} };

    \node (alg1) at (-0.5,-9) [draw,thick,minimum width=2cm,minimum height=1cm] {Algorithm~\ref{alg:solid-grid}};
        
    \node (lem19) at (4,-4) [draw,thick,minimum width=2cm,minimum height=1cm] {\begin{tabular}{c} Lemma~\ref{lem:corner-path-props}     
    \end{tabular} };

    \node (lem20) at (7,-4) [draw,thick,minimum width=2cm,minimum height=1cm] {\begin{tabular}{c} Lemma~\ref{lem:disjoint}     
    \end{tabular} };
    
    \node (lem21) at (10,-4) [draw,thick,minimum width=2cm,minimum height=1cm] {\begin{tabular}{c} Lemma~\ref{lem:cardinality}     
    \end{tabular} };

    \node (lem13) at (-10,-6) [draw,thick,minimum width=2cm,minimum height=1cm] {\begin{tabular}{c} Lemma~\ref{lem:exist-top}     
    \end{tabular} };

    \node (lem10) at (-7,-6) [draw,thick,minimum width=2cm,minimum height=1cm] {\begin{tabular}{c} Lemma~\ref{lem:intersect}     
    \end{tabular} };

    \node (prp11) at (-4,-6) [draw,thick,minimum width=2cm,minimum height=1cm] {\begin{tabular}{c} Proposition~\ref{rslt:cut-vertex-geod} \\
    \cite{ekim2014}
    \end{tabular} };

    \node (lem18) at (10,-6) [draw,thick,minimum width=2cm,minimum height=1cm] {\begin{tabular}{c} Lemma~\ref{lem:vertical-neighbour-path}     
    \end{tabular} };

    \node (obs5) at (-4.5,-9) [draw,thick,minimum width=2cm,minimum height=1cm] {\begin{tabular}{c} Observation~\ref{obs:cycle-square}     
    \end{tabular} };

    \node (obs6) at (3,-9) [draw,thick,minimum width=2cm,minimum height=1cm] {\begin{tabular}{c} Observation~\ref{obs:iso-path}     
    \end{tabular} };

    \node (obs8) at (-10,-9) [draw,thick,minimum width=2cm,minimum height=1cm] {\begin{tabular}{c} Observation~\ref{obs:boundary}
    
    \end{tabular} };

    % Arros from bottom to top

    \draw[->,thick,densely dotted] (obs5.north) ++ (-0.3,0) -- ++ (0,0.5) -- ++ (-7,0) |- (lem13.west);
    \draw[->,thick,densely dotted] (obs5.north) -- ++(0,0.5) -| (lem18.south);
    \draw[->,thick,densely dotted] (obs6.north) -- ++ (0,0.4) -- ++(0.2,0) -- ++ (0,0.6) -- ++ (-11.75,0) |- (lem15.west);
    \draw[->,thick,densely dotted] (obs6.north) (obs6.north) -- ++ (0,0.4) -- ++(0.2,0) -- ++ (0,0.6) -- ++(0,2.35) -| (lem21.south);
    
     \draw[->,thick,densely dotted] (obs8.west) -| (lem9.south);
    
     \draw[->,thick,densely dotted] (obs8.north) -- (lem13.south);
     \draw[->,thick] (obs8.east) -| (lem10.south);

    \draw[->,thick,densely dotted] (alg1.north) -- ++ (0,2) -| (obs7.south);
    % \draw[->,thick,densely dotted] (alg1.north) ++ (0,2) -| (lem23.south);
    \draw[->,thick,densely dotted] (alg1.north) ++ (0,1.25) -- ++(-10,0) -- ++ (0,0.75);
    
    \draw[->,thick,densely dotted] (lem18.north) -- (lem21.south);
    \draw[->,thick,densely dotted] (prp11.north) -- (obs12.south);
    \draw[->,thick,densely dotted] (lem10.north) -- (lem15.south);
    \draw[->,thick,densely dotted] (lem13.north) -- (lem14.south);

    \draw[->,thick,densely dotted] (lem9.north) |- (lem17.west);
    \draw[->,thick,densely dotted] (lem14.north) -- (lem17.south);
    \draw[->,thick,densely dotted] (lem15.north) -- (lem17.south);
    \draw[->,thick,densely dotted] (obs12.north) |- (lem17.east);
    
    \draw[->,thick,densely dotted] (lem19.north) |- (lem22.west);
    \draw[->,thick,densely dotted] (lem19.east) -- (lem20.west);

    \draw[->,thick,densely dotted] (lem20.north) -- (lem22.south);
    \draw[->,thick,densely dotted] (lem21.north) |- (lem22.east);
    
    \draw[->,thick,densely dotted] (lem22.north) |- (thm1.east);
    % \draw[->,thick,densely dotted] (lem23.north) -- (thm1.south);
    \draw[->,thick,densely dotted] (obs7.north) -- (thm1.south);
    \draw[->,thick,densely dotted] (lem17.north) |- (thm1.west);
    
    % \draw[->,thick,densely dotted] (obs12.north) |- (lem17.east);

\end{tikzpicture}}%

    \caption{Roadmap of proof of Theorem~\ref{thm:solid-grid}.}
    \label{fig:solid-grid-navigate}
\end{figure}

\subsection{Preliminaries}\label{subsec:solid-prelim}

Let $G$ be a solid grid and $\mathcal{R}$ its grid embedding.  In this section, whenever we refer to a grid embedding of a solid grid, we refer to a grid embedding whose each interior face has area exactly four. Each vertex of $G$ corresponds to a grid point in $\mathcal{R}$, and each edge of $G$ corresponds to an orthogonal unit segment. For an edge $e$, $\mathbf{e}$ shall denote the corresponding orthogonal unit segment in $\mathcal{R}$. For a subgraph $H$, let $\mathcal{R}_H = \bigcup\limits_{e\in E(H)} \mathbf{e}$. Observe that for a cycle $C$ of $G$, $\mathcal{R}_C$ is a simple, closed, rectilinear curve that divides the plane into two disjoint regions. Let $\interior{\mathcal{R}_C}$ denote the union of $\mathcal{R}_C$ and ``the interior region bounded by  $\mathcal{R}_C$''.  A square whose each side is an orthogonal segment containing exactly two grid points is referred to as a \textit{unit square}. Note that the length of a each side of a unit square is one. For a unit square $U$ whose all four corner points are grid points, let $\topleft{U}$, $\topright{U}$, $\bottomleft{U}$, $\bottomright{U}$ denote the top-left, top-right, bottom-left, and bottom-right corner point of $U$, respectively. We shall use the following observation which essentially says that there is no ``hole'' inside the grid embedding of a cycle of a solid grid graph.

\begin{observation}\label{obs:cycle-square}
      Let $G$ be a solid grid graph and $\mathcal{R}$ be a grid embedding of $G$. Let $C$ be a cycle of $G$ and $\mathcal{R}_C$ be the sub-embedding of $C$ in $\mathcal{R}$. Let $U$ be a unit square whose all four corner points are grid points and $U\subseteq \interior{\mathcal{R}_C}$. Then, the four corner points $\topleft{U}$, $\topright{U}$, $\bottomleft{U}$, $\bottomright{U}$ induce a cycle of order four in $G$.
\end{observation}
\begin{proof}
   If the four corner points $\topleft{U}$, $\topright{U}$, $\bottomleft{U}$, $\bottomright{U}$ do not induce a cycle of order four in $G$, then there would be an interior face in  $\mathcal{R}$ whose area would be more than four. This would contradict that $\mathcal{R}$ is a  grid embedding of $G$.
\end{proof}

% of $G$ where every interior face has unit area. 
%Let $G$ be a solid grid graph.

  Let $G$ be a solid grid graph and $\mathcal{R}$ be a grid embedding of $G$. A path $P$ in $G$ is a \emph{vertical path} (resp. \emph{horizontal path}) if $x$-coordinates (resp. $y$-coordinates) of all vertices of $P$ are the same. An \textit{isometric path} is a shortest path between its end vertices. We shall also use the following observation.

\begin{observation}\label{obs:iso-path}
  Let $G$ be a solid grid graph and $\mathcal{R}$ be a grid embedding of $G$. Let $P_1$ and $P_2$ be a vertical path and a horizontal path, respectively. Any path $P_3$ with $V(P_3) \subseteq V(P_1)\cup V(P_2)$ is an isometric path. 
\end{observation}

In the next section, we state our algorithm to compute the minimum geodetic set of a solid grid.

\subsection{Algorithm}\label{subsec:alg-solid}

   Let $G$ be a solid grid graph and $\mathcal{R}$ be a grid embedding of $G$. A vertical (resp. horizontal) path $P$ in $G$ with at least two vertices is a \emph{vertical corner path} (resp. \emph{horizontal corner path}) if (i) no vertex of $P$ is a \textit{cut-vertex}\footnote{A vertex of a graph is a \emph{cut vertex} if its removal disconnects the graph.}, (ii) both end-vertices of $P$ have degree~$2$ in $G$, (iii) all other vertices of $P$ have degree~$3$ in $G$. A path of $G$ is a \emph{corner path} if it is either a vertical corner path or a horizontal corner path. See Figure~\ref{fig:corner-path-example} for an example. A vertex $v \in V(G)$ is a \emph{corner vertex} if $v$ is an end-vertex of some corner path.

    \begin{figure}
    \centering
    \includegraphics[scale=1.5]{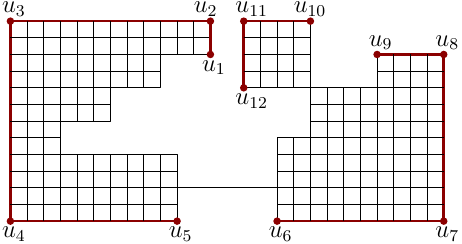}
    \caption{{ Illustration for corner paths. Here, the horizontal and vertical paths marked in bold red are corner paths. Moreover, $(u_1,u_2,u_3,u_4,u_5)$, $(u_6,u_7,u_8,u_9)$, and $(u_{10},u_{11},u_{12})$ are maximal corner sequences in this solid grid.}}
    \label{fig:corner-path-example}
\end{figure}

%All corner vertices can be found in linear time even if the grid embedding is not provided as an input~\cite{caldam2020}.

% \begin{proposition}[\cite{caldam2020}]\label{prp:solid}
% Given a solid grid graph, all corner vertices can be found in linear time.
% \end{proposition}
% A vertex $v$ of $G$ is a \emph{corner vertex} if $v$ has degree $1$ or $v$ is an end-vertex of some corner path. 

\begin{definition}
   Let $G$ be a solid grid graph and $\mathcal{R}$ be a grid embedding of $G$.	We say that vertices $u_1,u_2,\ldots,u_k$ form a \emph{corner sequence} if for each $1\leq i\leq k-1$,
	there is a corner path with $u_i$ and $u_{i+1}$ as end-vertices. A corner sequence is \emph{maximal} if it is not a subsequence of any other corner sequence. Two corner sequences are \emph{distinct} from each other if neither is a permutation of the other. For a corner sequence $S$, let $|S|$ denote the length of $S$.
	%\begin{itemize}
% 	1.  and\\
% 	2. there is no corner vertex in the clockwise traversal of the boundary of $\mathcal{R}$ from $u_i$ to $u_{i+1}$.
	%	\end{itemize}
\end{definition}

 \begin{minipage}{\linewidth}
 \begin{algorithm}[H] 
 \DontPrintSemicolon
\begin{small}
   \SetKwFunction{compute}{Compute\_sCD(G,a,A)}
    \SetKwInOut{KwIn}{Inputs}
    \SetKwInOut{KwOut}{Output}
    \KwIn{A solid grid graph $G$ and its embedding $\mathcal{R}$}
    \KwOut{A minimum geodetic set of $G$}
    
    Let $V_1$ denote the set of all vertices of degree $1$;
    
    Find all corner paths and corner vertices.
    
    Let $\mathcal{S}$ be the set of all distinct maximal corner sequences.
    
    Let $D=V_1$.
    
    \For{each {distinct} maximal sequence $S\in \mathcal{S}$}{
        Let $S=(u_1,u_2,\ldots,u_k)$.
        
        Let $f(S) = \{u_2,u_4\ldots,u_{k-k'}\}$ where $k'=0$ if $k$ is even and $k'=1$, otherwise.
        
        $D=D\cup f(S)$.
    }
 \Return $D$.
   \caption{Pseudocode to find a minimum geodetic set of a solid grid graph}
   \label{alg:solid-grid}
    
\end{small}
 \end{algorithm}
\end{minipage}

\begin{figure}
    \centering
    \includegraphics[scale=0.87]{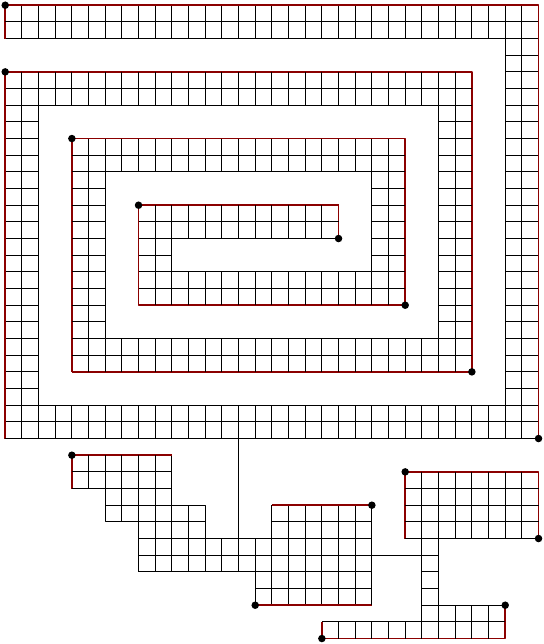}
    \caption{Illustration for execution of Algorithm~\ref{alg:solid-grid} on a solid grid. The marked vertices are selected by the algorithm. Moreover, the corner paths are represented by bold red paths.}
    \label{fig:example}
\end{figure}

 \medskip \noindent Figure~\ref{fig:example} show an example of a solid grid $G$ and the set of vertices returned by Algorithm~\ref{alg:solid-grid}. Since the embedding is given as part of the input, the running time is $O(n)$, where $n$ is part of the input. The next observation follows immediately from the above algorithm. 

\begin{observation}\label{obs:caardinality-solid}
    Let $D$ be the set returned by Algorithm~\ref{alg:solid-grid}. Let $t$ be the number of vertices of $G$ with degree one and $\mathcal{S}$ set of all {distinct} maximal corner sequences. Then, $|D| \leq t + \sum_{S\in \mathcal{S}} \left\lfloor|S|/2\right\rfloor$. 
\end{observation}

In the next section, we show that the set returned by Algorithm~\ref{alg:solid-grid} is indeed a geodetic set of the input solid grid.

\subsection{Feasibility}\label{subsec:feasible}

In this section, we shall show that Algorithm~\ref{alg:solid-grid} returns a geodetic set. In Section~\ref{subsec:boundary}, we shall prove some properties of the boundary of a biconnected component of a solid grid. Finally, in Section~\ref{subsubsec:feasible}, we shall use the developed machinery to establish the feasibility of our algorithm.

\subsubsection{Properties of the boundary of a biconnected component of a solid grid}\label{subsec:boundary}

\newcommand{\northv}[1]{north\left(#1\right)}
\newcommand{\southv}[1]{south\left(#1\right)}
\newcommand{\eastv}[1]{east\left(#1\right)}
\newcommand{\westv}[1]{west\left(#1\right)}
\newcommand{\Topchain}[1]{Tchain\left(#1\right)}
\newcommand{\Bottomchain}[1]{Bchain\left(#1\right)}
\newcommand{\Leftchain}[1]{Lchain\left(#1\right)}
\newcommand{\Rightchain}[1]{Rchain\left(#1\right)}
\newcommand{\TopLeftchain}[1]{TLchain\left(#1\right)}
\newcommand{\TopRightchain}[1]{TRchain\left(#1\right)}
\newcommand{\BottomLeftchain}[1]{BLchain\left(#1\right)}
\newcommand{\BottomRightchain}[1]{BRchain\left(#1\right)}

Let $G$ be a solid grid graph and $\mathcal{R}$ be a grid embedding of $G$. Let $H$ be a biconnected component of $G$. For a vertex $v$ of $H$, let $X_v$ denote the maximal horizontal path of $H$ containing $v$ and $Y_v$ be the maximal vertical path of $H$ containing $v$. Let $\northv{v}$ and $\southv{v}$ denote the end-vertices of $Y_v$ that have the highest and lowest $y$-coordinates, respectively. 
%Let $\southv{v}$ denote the end-vertex of $Y_v$ that has the lowest $y$-coordinate. 
Similarly, let $\eastv{v}$ and $\westv{v}$ denote the end-vertices of $X_v$ that have the highest and lowest $x$-coordinates, respectively. 
%Let $\westv{v}$ denote the end-vertex of $X_v$ that has the lowest $x$-coordinate. 
Observe that the vertices $\northv{v},\westv{v},\southv{v},\eastv{v}$ lie on the boundary of $\mathcal{R}_H$. Moreover, if $v$ does not lie in the boundary of $\mathcal{R}_H$,  then $v,\northv{v},\westv{v},\southv{v},\eastv{v}$ are all distinct vertices. 

 Let $G$ be a solid grid graph and $\mathcal{R}$ be a grid embedding of $G$. Let $H$ be a biconnected component of $G$. For a vertex $v$ of $H$, let $\Topchain{v}$ denote the subgraph induced by the set of vertices encountered in a counter-clockwise traversal of the boundary of $\mathcal{R}_H$ from $\eastv{v}$ to $\westv{v}$. Let $\Bottomchain{v}$ denote the subgraph induced by the set of vertices encountered in a counter-clockwise traversal of the boundary of $\mathcal{R}_H$ from $\westv{v}$ to $\eastv{v}$. Let $\Leftchain{v}$ denote the set of vertices encountered in a counter-clockwise traversal of the boundary of $\mathcal{R}_H$ from $\northv{v}$ to $\southv{v}$. Let $\Rightchain{v}$ denote the set of vertices encountered in a counter-clockwise traversal of the boundary of $\mathcal{R}_H$ from $\southv{v}$ to $\northv{v}$. See Figure~\ref{fig:example-ChainPaths} for an illustration.

\begin{figure}
    \centering
    \includegraphics[scale=1]{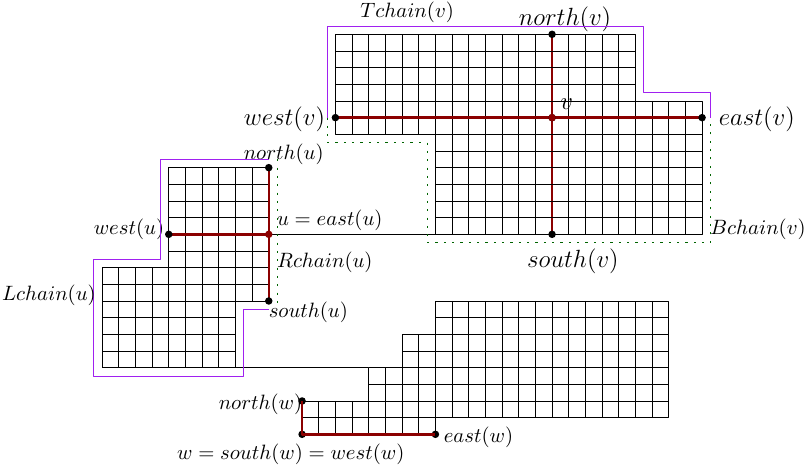}
    \caption{For vertex $v$, we outline $Tchain(v)$ with a purple line and $Bchain(v)$ with a green dotted line. Similarly, for vertex $u$, we outline $Lchain(u)$ with a purple line and $Rchain(u)$ with a dotted green line.}
    \label{fig:example-ChainPaths}
\end{figure}

 A biconnected component $H$ of $G$ is \emph{trivial} if $H$ is either just an edge or a single vertex. For all non-trivial biconnected components $H$ of $G$, we have the following observation.

 \begin{observation}\label{obs:boundary}
     Let $G$ be a solid grid graph and $\mathcal{R}$ be a grid embedding of $G$. Let $H$ be a non-trivial biconnected component of $G$. For a vertex $v$ of $H$, all the following hold.
     
     \begin{enumerate}[label=(\alph*)]

         \item\label{it:not-same-v} The vertices $\northv{v},\southv{v}$ are distinct and the vertices $\eastv{v}, \westv{v}$ are distinct.
         
         \item\label{it:same-v} Let $a\in \{ \northv{v},\southv{v} \}$ and $b\in \{ \eastv{v},\westv{v} \}$ be two vertices such that $a=b$. Then, $a=b=v$. 
         
         \item The subgraphs $\Topchain{v}$, $\Bottomchain{v}, \Leftchain{v}, \Rightchain{v}$ are paths in $G$.

         \item\label{it:path-contain} The path $\Topchain{v}$ contains $\northv{v}$ and the path $\Bottomchain{v}$ contains $\southv{v}$.
         
         \item\label{it:cycle-1} The union of the edges of $\Leftchain{v} $ and $ \Rightchain{v}$ creates a cycle $C$ in $G$ and $\mathcal{R}_C$ is the boundary of $\mathcal{R}_H$.

         \item\label{it:cycle-2} The union of the edges of $\Topchain{v}$ and $\Bottomchain{v}$ creates a cycle $C$ in $G$ and $\mathcal{R}_C$ is the boundary of $\mathcal{R}_H$.

         \item\label{it:inside-path} Let $P$ be a path such that $V(P) \subseteq V(H)$. Then $\mathcal{R}_P \subset \interior{\mathcal{R}_C}~{\cup \mathcal{R}_C}$ where $C$ is the cycle corresponding to the boundary of $\mathcal{R}_H$. 
         \end{enumerate}
         
 \end{observation}

  Let $\TopLeftchain{v}$ be the subpath of $\Topchain{v}$ between $\northv{v}$ and $\westv{v}$. Similarly, let $\TopRightchain{v}$ be the subpath of $\Topchain{v}$ between $\northv{v}$ and $\eastv{v}$. Let $\BottomLeftchain{v}$ be the subpath of $\Bottomchain{v}$ between $\southv{v}$ and $\westv{v}$. Similarly, let $\BottomRightchain{v}$ be the subpath of $\Bottomchain{v}$ between $\southv{v}$ and $\eastv{v}$. Due to Observation~\ref{obs:boundary}, the above terminologies are well defined. Observe that $\Leftchain{v}$ is the concatenation of $\TopLeftchain{v}$ with $\BottomLeftchain{v}$ and $\Rightchain{v}$ is the concatenation of $\BottomRightchain{v}$ with $\TopRightchain{v}$. We have the following lemma.

 \begin{lemma}\label{lem:intersect-opposite}
    Let $G$ be a solid grid graph and $\mathcal{R}$ be a grid embedding of $G$. Let $H$ be a non-trivial biconnected component of $G$. For a vertex $v$ of $H$, we have $V(\TopRightchain{v}) \cap V(\BottomLeftchain{v}) \subseteq \{v\}$ and $V(\TopLeftchain{v}) \cap V(\BottomRightchain{v}) \subseteq \{v\}$. 
 \end{lemma}

\begin{proof}
    Due to symmetry, it is enough to prove that $V(\TopRightchain{v}) \cap V(\BottomLeftchain{v}) \subseteq \{v\}$. If each path in $\TopLeftchain{v}, \TopRightchain{v}, \BottomLeftchain{v}, \BottomRightchain{v}$ has at least one edge, {then it follows from the definitions that $V(\TopRightchain{v}) \cap V(\BottomLeftchain{v}) = \emptyset$.} Moreover, since $\northv{v} \neq \southv{v}$ and $\eastv{v} \neq \westv{v}$  (due to Observation~\ref{obs:boundary}\ref{it:not-same-v}), we have to consider only the following four exhaustive cases.

    First, suppose $\southv{v}=\westv{v}$. Due to Observation~\ref{obs:boundary}\ref{it:same-v} we have that $\southv{v}=\westv{v} = v$, and therefore, $V(\BottomLeftchain{v})=\{v\}$. This concludes the proof of the lemma for this case.
    
    % Now using Observation~\ref{obs:boundary}\ref{it:not-same-v}, we infer that $\northv{v} \neq v$ and $\eastv{v} \neq v$. This implies that $v\notin \TopRightchain{v}$ and therefore $V(\TopRightchain{v}) \cap V(\BottomLeftchain{v}) = \emptyset$. 

 \sloppy Second, suppose $\southv{v}=\eastv{v}$. Due to Observation~\ref{obs:boundary}\ref{it:same-v} we have that $\southv{v}=\eastv{v} = v$, and therefore, $\BottomLeftchain{v}=\Bottomchain{v}$. From the definitions it is clear that $\Topchain{v} \cap \Bottomchain{v} = \{\eastv{v},\westv{v}\}$. Using Observation~\ref{obs:boundary}\ref{it:not-same-v}, we infer that $\northv{v} \neq v$ and therefore, $\TopRightchain{v}$ is proper subpath of the path induced by $\Topchain{v}$ and does not contain $\westv{v}$. This implies $V(\TopRightchain{v}) \cap V(\BottomLeftchain{v}) \subseteq \{\eastv{v}\}=\{v\}$. 

 As the cases $\northv{v}=\eastv{v}$ and $\northv{v}=\westv{v}$ are symmetrical to the above cases, we have the proof of the Lemma. \end{proof}
    % Third, suppose $\northv{v}=\eastv{v}$. Due to Observation~\ref{obs:boundary}\ref{it:same-v} we have that $\northv{v}=\eastv{v} = v$ and therefore $V(\TopRightchain{v})=\{v\}$. Now using Observation~\ref{obs:boundary}\ref{it:not-same-v} we infer that $\southv{v} \neq v$ and $\westv{v} \neq  v$. 
    % This implies that $v\notin \BottomLeftchain{v}$ and therefore $V(\TopRightchain{v}) \cap V(\BottomLeftchain{v}) = \emptyset$. 

    % Finally, suppose $\northv{v}=\westv{v}$. Due to Observation~\ref{obs:boundary}\ref{it:same-v} we have that $\northv{v}=\westv{v} = v$ and therefore $\TopRightchain{v}=\Topchain{v}$. Now using Observation~\ref{obs:boundary}\ref{it:not-same-v} we infer that $\southv{v} \neq v$ and $\eastv{v} \neq v$. Now Observation~\ref{obs:boundary}\ref{it:path-contain}, \ref{it:cycle-1} and \ref{it:cycle-2} together imply $V(\TopRightchain{v}) \cap V(\BottomLeftchain{v}) = \emptyset$.    

\begin{lemma}\label{lem:intersect}
    Let $G$ be a solid grid graph and $\mathcal{R}$ be a grid embedding of $G$. Let $H$ be a non-trivial biconnected component of $G$. Let $v$ be any vertex of $H$ and $u,w$ be two vertices such that $u\in V(\TopRightchain{v})$ and $w\in V(\BottomLeftchain{v})$. Let $P$ be any path between $u$ and $w$ in $H$. Then $V(P)\cap V(X_v)\neq \emptyset$ and $V(P)\cap V(Y_v)\neq \emptyset$.
\end{lemma}
\begin{proof}
    First, we show that $V(P) \cap V(X_v) \neq \emptyset$. If $\{u,w\}\cap V(X_v) \neq \emptyset$, then we are done. Let $C$ denote the subgraph induced by $V(X_v) \cup V(\Topchain{v})$. Observe that if no vertex of $V(X_v) \setminus \{\westv{v},\eastv{v}\}$ lie on the boundary of $\mathcal{R}_H$, then $C$ is a cycle. Otherwise, $C$ can be decomposed into cycles $C_1,C_2,\ldots,C_t$ and paths $Q_1,Q_2,\ldots,Q_{t'}$.  Since $\TopRightchain{v}\subseteq \Topchain{v}$ and $\BottomLeftchain{v} \subseteq \Bottomchain{v}$, $u\in \Topchain{v}$ and $w\in \Bottomchain{v}$. Hence $w\notin \displaystyle \bigcup\limits_{i=1}^{t} \interior{\mathcal{R}_{C_i}} \displaystyle \bigcup\limits_{i=1}^{t'} \interior{\mathcal{R}_{Q_i}}$ and therefore, $\mathcal{R}_P \not\subseteq \displaystyle \bigcup\limits_{i=1}^{t} \interior{\mathcal{R}_{C_i}} \displaystyle \bigcup\limits_{i=1}^{t'} \interior{\mathcal{R}_{Q_i}}$. 
    % Hence $\mathcal{R}_P$ and $\mathcal{R}_{C}$ must have at least one point in common. 
    {Let $p$ be the point closest to $w$ (in $\mathcal{R}_P$) that also lies on  $\mathcal{R}_{C}$. If $p \in \Topchain{v}$, then it follows that $\mathcal{R}_P \not\subseteq \interior{\mathcal{R}_{C'}}~\cup \mathcal{R}_{C'}$ where $C'$ is the boundary of $H$. But this would contradict Observation~\ref{obs:boundary}\ref{it:inside-path}. Hence we have that, $p \notin \Topchain{v}$ and it must be the case that $p\in \mathcal{R}_{X_v}$ and therefore $V(P) \cap V(X_v) \neq \emptyset$.} Through symmetric arguments, it can be shown that $V(P) \cap V(Y_v) \neq \emptyset$.
\end{proof}

\subsubsection{Algorithm~\ref{alg:solid-grid} returns a geodetic set}\label{subsubsec:feasible}

The reader may observe that Algorithm~\ref{alg:solid-grid} never picks a cut vertex in its solution. To prove the rationale behind it, we shall use the following result of Ekim and Erey~\cite{ekim2014}.

\begin{proposition}[\cite{ekim2014}]\label{rslt:cut-vertex-geod}
	Let $F$ be a graph and $F_1,\ldots,F_k$ its biconnected components. Let $C$ be the set of cut vertices of $G$. If $X_i\subseteq V(F_i)$ is a minimum set such that $X_i\cup (V(F_i)\cap C)$ is a minimum geodetic set of $F_i$, then $\cup_{i=1}^{k} X_i$ is a minimum geodetic set of $F$.
\end{proposition}

% Now we shall fix some notations that will remain valid for the remainder of this section. Let $\mathcal{S}$ be the set of all maximal corner sequences of $G$ and $t$ be the number of vertices of degree one. For a maximal corner sequence $S=u_1,u_2,\ldots,u_k$ let $f(S)$ denote the set $\{u_2,u_4\ldots,u_{k-k'}\}$ where $k'=0$ if $k$ is even and $k'=1$, otherwise. Observe that $|f(S)|=\left\lfloor\frac{k}{2}\right\rfloor$. Let $V_1$ be the set of all vertices of degree $1$. Now consider the sets $V_2=\cup_{S\in \mathcal{S}} f(S)$ and  $D=V_1\cup V_2$. Observe that the set returned by Algorithm~\ref{alg:solid-grid} is same as $D$. 

% Observe that $|D|=t + \sum_{S\in \mathcal{S}} \left\lfloor|S|/2\right\rfloor$. Therefore, once we show that $D$ is indeed a geodetic set of $G$, we will be done by Lemma~\ref{lem:cardinlaity-2}.

% Now we will show that $D$ is a geodetic set of $G$. 

The next observation follows from Proposition~\ref{rslt:cut-vertex-geod}.

\begin{observation}\label{obs:block}
Let $C(G)$ be the set of cut-vertices of $G$ and let $\{H_1,H_2,\ldots,H_k\}$ be the set of biconnected components of $G$. A set $D$ is a geodetic set of $G$ if and only if $(D\cup C(G))\cap V(H_i)$ is a geodetic set of $H_i$ for all $1\leq i\leq k$.
\end{observation}

\begin{lemma}\label{lem:exist-top}
    Let $G$ be a solid grid graph and $\mathcal{R}$ be a grid embedding of $G$. Let $C(G)$ be the set of cut-vertices of $G$ and $D$ be the set returned by Algorithm~\ref{alg:solid-grid}. Let $H$ be a non-trivial biconnected component of $G$ and $v$ be any vertex of $H$. Then, for each $Z\in \left\{ \Leftchain{v}, \Rightchain{v}, \Topchain{v}, \Bottomchain{v} \right\}$ we have that $Z \cap (D\cup C(G)) \neq \emptyset$.
\end{lemma}

\begin{proof}
Let $D_H$ denote the set $(D\cup C(G))\cap V(H)$.  Due to symmetry, it is sufficient to prove the lemma for $Z=\Topchain{v}$. Let $m$ be a vertex of $Z$ whose $y$-coordinate is maximum among all vertices of $Z$, and $X_{m}$ be the maximal horizontal path of $H$ containing $m$. Observe that $X_m$ is a subpath of $\Topchain{v}$. We shall show that either $X_m$ contains a cut vertex or it is a corner path.

Let $X_m$ do not contain any cut vertex and let $a$ and $b$ be the end-vertices of $X_m$. Suppose there exists a vertex $u\in \{a,b\}$ that has degree $3$ in $G$. Then, there exists a vertex $u' \in N[u]$, which is not a vertex of $H$. Hence $u$ is a cut-vertex, a contradiction. Therefore we have the following.

\medskip\noindent\textbf{(+)} Degrees of both $a$ and $b$ in $G$ are two.

Similarly, if any vertex $w$ of $X_m$ has a neighbour $w'$ whose $y$-coordinate is greater than that of $w$, then $w$ is also a cut-vertex which is a contradiction. Now, consider any vertex $c$ of $X_m$ distinct from both $a$ and $b$. Let the coordinate of $c$ be $(\ell,\ell')$. Consider the unit square $U$ that has $(\ell,\ell')$ and $(\ell-1,\ell'-1)$ as diagonally opposite points. Also consider the subgraph $F$ formed by the vertices of $V(\Topchain{v}) \cup V(\Bottomchain{v})$. Due to Observation~\ref{obs:boundary}\ref{it:cycle-1}, $F$ is a cycle of $G$ and observe that $\interior{\mathcal{R}_F}$ contains $U$. Due to Observation~\ref{obs:cycle-square}, the four corner points of $U$ induce a cycle in $G$. Therefore, we have the following.

\medskip\noindent\textbf{(++)} For any vertex of $c\in V(X_m)\setminus \{a,b\}$, degree of $c$ in $G$ is three.

Due to $(+)$ and $(++)$, we have that $X_m$ is a corner path. Now recall the definition of maximal corner sequence and observe that there exists a maximal corner sequence $S$ where $a$ and $b$ are consecutive. From the definition of $D$, we have that at least one vertex among $a$ and $b$ lies in $D$ and, therefore, in $D_H$. This completes the proof.
\end{proof}

\begin{lemma}\label{lem:exist-cov-1}
    Let $G$ be a solid grid graph and $\mathcal{R}$ be a grid embedding of $G$. Let $C(G)$ be the set of cut-vertices of $G$ and $D$ be the set returned by Algorithm~\ref{alg:solid-grid}. Let $H$ be a non-trivial biconnected component of $G$, and $v$ be any vertex of $H$. Then, at least one of the following is true. 
    
    \begin{enumerate}[label=(\alph*)]
        \item\label{it:exist-cov-1-a} There exist two vertices $\{u,w\} \subseteq (D\cup C(G))\cap V(H)$ such that $u\in V(\TopRightchain{v})$ and $w\in V(\BottomLeftchain{v})$.

        \item\label{it:exist-cov-1-b} There exist two vertices $\{u,w\} \subseteq (D\cup C(G))\cap V(H)$ such that $u\in V(\TopLeftchain{v})$ and $w\in V(\BottomRightchain{v})$.
        
    \end{enumerate} 
\end{lemma}

\begin{proof}
Let $D_H=(D\cup C(G))\cap V(H)$. Due to Lemma~\ref{lem:exist-top}, we know that there exists at least one vertex $u\in D_H$ such that $u\in V(\Bottomchain{v})$. Due to symmetry, we assume that $u\in \BottomLeftchain{v}$. Now suppose \ref{it:exist-cov-1-a} is false. Then $D_H\cap \TopRightchain{v}=\emptyset$. Applying Lemma~\ref{lem:exist-top}, we infer that there exists a vertex $u' \in D_H\cap \Topchain{v}$ and therefore $u'\in \TopLeftchain{v}$. We apply Lemma~\ref{lem:exist-top} again to infer that there exists a vertex $w' \in D_H\cap \Rightchain{v}$. Since $\Rightchain{v}$ is the concatenation of $\BottomRightchain{v}$ with $\TopRightchain{v}$, we infer that $v'$ must lie in $\BottomRightchain{v}$. Hence \ref{it:exist-cov-1-b} is true.
\end{proof}

\begin{lemma}\label{lem:cov-1}
    Let $G$ be a solid grid graph and $\mathcal{R}$ be a grid embedding of $G$. Let $H$ be a non-trivial biconnected component of $G$. Let $v$ be any vertex of $H$ and $u,w$ be two vertices such that $u\in V(\TopRightchain{v})$ and $w\in V(\BottomLeftchain{v})$. Then there exists a shortest path between $u$ and $w$ that contains $v$.
\end{lemma}

\begin{proof}
    Let $P$ be a shortest path between $u$ and $w$. If $v\in V(P)$, we are done. Recall that $X_v$ (resp. $Y_v$) is the maximal horizontal (resp. vertical) path of $H$ containing $v$. Lemma~\ref{lem:intersect} implies that $V(P)\cap X_v \neq \emptyset$ and $V(P)\cap Y_v \neq \emptyset$. Let $P$ be be written as $u=z_0~z_1~z_2~\ldots z_t=w $. Let $i$ be the minimum index such that $z_i \in V(X_v)\cup V(Y_v)$ and $Q\in \{X_v,Y_v\}$ such that $z_i\in V(Q)$. Since $z_i\neq v$, there exists a maximum index $j$, distinct from $i$, and a path $Q'\in \{X_v,Y_v\}$ distinct from $Q$ such that $z_j\in V(Q')$. Now, consider the path $P'$ between $z_i$ and $z_j$ such that $V(P') \subseteq V(X_v)\cup V(Y_v)$. Observation~\ref{obs:iso-path} implies that $P'$ is an isometric path. Hence, the path $Q=z_0 \ldots z_{i-1}~P'~z_{j+1}\ldots z_t$ is a also an isometric path that contains $v$.
    % Therefore, we can replace $z_i,z_j$-subpath of $P$ with $P'$ to get an isometric $u,w$-path containing $v$.
    % First, consider the case when both $z_i$ and $z_j$ lie on $X_v$ (resp. $Y_v$). Now observe the horizontal (resp. vertical) subpath of $X_v$ (resp. $Y_v$), say $Q_1$, between $z_i$ and $z_j$. Due to Observation~\ref{obs:iso-path}, $Q_1$ is an isometric path and it contains $v$.  
    % Now consider the case when $z_i\in V(Y_v), z_j \in V(X_v)$ (resp. $z_i\in V(Y_v), z_j \in V(X_v)$ ). Now observe the path, say $Q_2$, between $z_i$ and $z_j$ whose vertex set is a subset of $X_v$ and $Y_v$.  Due to Observation~\ref{obs:iso-path}, $Q_2$ is an isometric path and it contains $v$. Hence the path $Q=z_0 \ldots z_{i-1}~Q_2~z_{j+1}\ldots z_t$ is a also an isometric path that contains $v$, and we are done.
\end{proof}

Due to symmetry, we also have the following lemma.

\begin{lemma}\label{lem:cov-2}
    Let $G$ be a solid grid graph and $\mathcal{R}$ be a grid embedding of $G$. Let $H$ be a non-trivial biconnected component of $G$. Let $v$ be any vertex of $H$ and $u,w$ be two vertices such that $u\in V(\TopLeftchain{v})$ and $w\in V(\BottomRightchain{v})$. Then there exists a shortest path $P$ between $u$ and $w$ such that $v\in V(P)$.
\end{lemma}

\begin{lemma}\label{lem:feasibility}
    Algorithm~\ref{alg:solid-grid} returns a geodetic set of the input solid grid graph.
\end{lemma}

\begin{proof}
    Let $G$ be a solid grid graph and $\mathcal{R}$ be a grid embedding of $G$. Let $D$ be the set returned by Algorithm~\ref{alg:solid-grid} with $G$ as input. Let $H$ be any biconnected component of $G$. If $H$ is an edge, then either both vertices of $H$ are cut vertices, or at least one of the vertices is degree one. In either case, both vertices lie in some shortest path between some pair of vertices in $D$. So, we assume $H$ is a non-trivial biconnected component and $v$ to be any vertex of $H$. Then, due to Lemma~\ref{lem:exist-cov-1}, there exists two vertices $\{u,w\} \subseteq (D\cup C(G))\cap V(H)$ such that either $(i)\ u\in V(\TopRightchain{v})$ and $w\in V(\BottomLeftchain{v})$, or $(ii)$ $u\in V(\TopLeftchain{v})$ and $w\in V(\BottomRightchain{v})$. In either case, if $u=w$ then due to Lemma~\ref{lem:intersect-opposite}, we have that $u=w=v$. Otherwise, due to Lemma~\ref{lem:cov-1} and Lemma~\ref{lem:cov-2}, we have that $v$ lies in a shortest path between $u$ and $w$. Hence, $(D\cup C(G))\cap V(H)$ is a geodetic set of $H$. Now arguing for all biconnected components of $G$ as above, and due to Observation~\ref{obs:block}, we have that $D$ is a geodetic set of $G$.
\end{proof}
{
\begin{lemma}\label{lem:edgeFeasibe}
    Algorithm~\ref{alg:solid-grid} returns an edge geodetic set of the input solid grid graph.
\end{lemma}
\begin{proof}
Targeting contradiction, assume that the set $S$ returned by Algorithm~\ref{alg:solid-grid} is not an edge geodetic set of the input solid grid graph. Then, there is an edge $uv$ that is not covered by $S$. Let $(x_u,y_u)$ and $(x_v,y_v)$ be the coordinates of $u$ and $v$ in the grid embedding $\mathcal{R}$. Without loss of generality, assume $y_u=y_v$ and that $x_v=x_u+1$. 

Recall that $S$ is a geodetic set (due to Lemma~\ref{lem:feasibility}). Hence, if $uv$ is a cut edge, then $S$ covers $uv$, and we are done. Now assume $uv$ belongs to some non-trivial biconnected component $H$. Due to Lemma~\ref{lem:exist-cov-1} and without loss of generality assume, there exist two vertices $\{w,w'\} \subseteq S\cap V(H)$ such that  $w\in V(\TopLeftchain{u})$ and  $w'\in V(\BottomRightchain{u})$. Observe that if $w'\in V(\BottomRightchain{v})$ as well, then the edge $uv$ lies in some shortest path between $w$ and $w'$. Otherwise, $w'$ must be the bottom-most endpoint of the maximal vertical path passing through $u$, i.e., $Y_u$.

Again applying Lemma~\ref{lem:exist-cov-1}, we know that either there exists  a vertex $\{w''\} \subseteq S\cap V(H)$ such that  $w''\in V(\TopRightchain{v})$ or there exists some vertices $\{z,z'\}\subseteq S\cap V(H)$ such that  $z\in V(\TopLeftchain{u})$ and  $z'\in V(\BottomRightchain{u})$. Now observe that if $w''$ exists, then the edge $uv$ lies in some shortest path between $w'$ and $w''$. Otherwise, the edge $uv$ lies in some shortest path between $w$ and $z'$.
\end{proof}
}

\subsection{Optimality}\label{sec:opt}

In this section, we shall prove that Algorithm~\ref{alg:solid-grid} indeed returns a minimum geodetic set. In the next section, we shall prove some properties of the corner paths. In Section~\ref{subsec:opt} we shall establish a lower bound on the geodetic number of a solid grid. 

\subsubsection{Properties of corner paths}

In the following lemma, we prove a structural property of the subgraph induced by the neighbours of the vertices of a corner path $P$ that lies outside $P$. 
% The following lemma was proved by Chakraborty et al~\cite{caldam2020} using a slightly different notion of corner paths. We provide its proof for completeness.

\begin{lemma}\label{lem:vertical-neighbour-path}
  Let $G$ be a solid grid graph and $\mathcal{R}$ be a grid embedding of $G$. Let $P$ be a vertical (resp. horizontal) corner path  of $G$ such that the $x$-coordinates (resp. $y$-coordinates) of all vertices of $P$ is $\ell$. Consider the set $Q = \{ v\in V(G)\colon v\notin V (P ), N (v) \cap P \neq \emptyset\}$. Then the following holds:

\begin{enumerate}[label=(\alph*)]
    \item\label{it:a} for a vertex $u\in P$, there exists a unique neighbour $u'\in Q$ and $y$-coordinates (resp. $x$-coordinates) of both $u,u'$ are same;
    
    \item\label{it:b} $x$-coordinates (resp. $y$-coordinates) of all vertices of $Q$ are the same and are equal to either $\ell-1$ or $\ell+1$;
     
    \item\label{it:c} $Q$ induces a path in $G$.
\end{enumerate} 
\end{lemma}

\begin{proof}
We shall prove the observation only for the case when $P$ is a vertical corner path, as the other case is symmetric. 

First, we prove~\ref{it:a}. Let $u$ be a vertex of $P$. Note that every internal vertex of $P$ must have degree three, and both end vertices of $P$ have degree two in $G$. Therefore, $u$ has exactly one neighbour in $Q$. As $P$ is a vertical path, the neighbour of $u$ that is not in the path must have the same $y$-coordinate as $u$. Moreover, since any two distinct vertices of $P$ have distinct $y$ coordinates, they have distinct neighbours in $Q$, as well.

\begin{figure}
    \centering
    \includegraphics[scale=0.5]{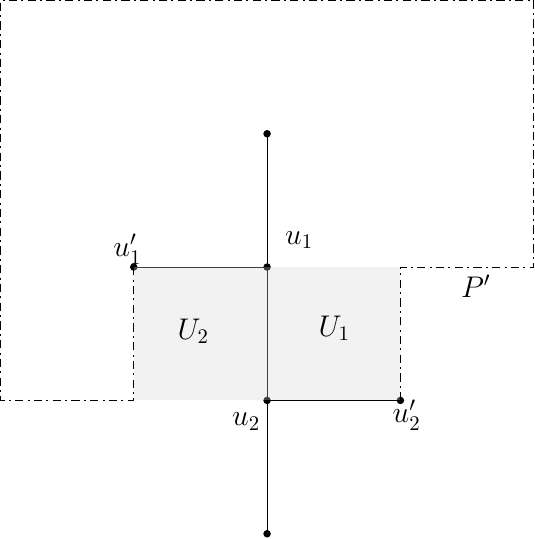}
    \caption{Illustration of the notations used in proof of Lemma~\ref{lem:vertical-neighbour-path}\ref{it:b}.}
    \label{fig:illus-1}
\end{figure}
Now we prove \ref{it:b}. Refer to Figure~\ref{fig:illus-1} for an illustration of the notations used in this proof. For a vertex $u$ of $P$, let $u'$ denote its unique neighbour in $Q$. Suppose the statement is false, then we deduce that there must exist two vertices $u_1$, $u_2$ of $P$ such that they are consecutive in $P$, the $x$-coordinate of $u_1'$ is $\ell-1$ and the $x$-coordinate of $u_2'$ is $\ell+1$. Without loss of generality, we assume that the $y$-coordinate of $u_1$ is larger than that of $u_2$ and equals to, say, $\ell_y$. Hence,  $u_1=(\ell,\ell_y)$, $u_2=(\ell,\ell_y-1)$, $u'_1=(\ell-1,\ell_y)$, and $u'_2=(\ell+1,\ell_y-1)$. Let $U_1$ denote the unit square formed by the grid points $\topleft{U_1}=u_1$, $\bottomleft{U_1}=u_2$, $\bottomright{U_1}=u'_2$ and $\topright{U_1}=(\ell+1,\ell_y)$. Let $U_2$ denote the unit square formed by the grid points $\topleft{U_2}=u'_1$, $\bottomleft{U_2}=(\ell-1,\ell_y-1)$, $\bottomright{U_2}=u_2$ and $\topright{U_2}=u_1$. Since $u_2$ is not a cut-vertex, there must exist a path $P'$ between $u_1'$ and $u_2'$ that does not contain $u_2$. Then the subgraph $H$ formed by $E(P')\cup \{u_1u'_1,u_1u_2,u_2u'_2\}$ contains a cycle $C$ such that there exists a $Z\in \{U_1,U_2\}$ with $Z\subseteq \interior{\mathcal{R}_C}$. Due to Observation~\ref{obs:cycle-square}, the four corner points of $Z$ induce a cycle of order four in $G$. If $Z=U_1$, then $u_1$ must have two distinct neighbours in $Q$, which is a contradiction. Similarly, if $Z=U_2$, then $u_2$ must have two distinct neighbours in $Q$, which is also a contradiction. 

Finally we prove \ref{it:c}. Let $P=u_1~u_2\ldots~u_k$ and $Q=\{u'_1,  u'_2, \ldots, u'_k\}$ such that $u'_i$ is the unique neighbour of $u_i$ in $N(u_i)\setminus V(P)$. Due to \ref{it:b}, the $x$-coordinates of $Q$ are all same. Without loss of generality, we assume that the $x$-coordinates of $Q$ are $\ell+1$. Let $i\in [k-1]$ and consider the vertices $u'_i$ and $u'_{i+1}$. Let $X$ denote the unit square formed by the grid points $u_{i+1}, u_{i}, u'_{i}, u'_{i+1}$. Since $u_{i+1}$ is not a cut vertex there exists a path $P'$ from $u'_{i+1}$ to $u_i$ that does not contain $u_{i+1}$. Now consider the subgraph $H$ formed by the edges $E(P')\cup \{u'_{i+1}u_{i+1},u_{i+1}u_{i},u_{i}u'_{i}\}$. The subgraph $H$ contains a cycle $C$ such that $X\subseteq \interior{\mathcal{R}_C}$. Then due to Observation~\ref{obs:cycle-square} we have that $u_{i+1}, u_{i}, u'_{i}, u'_{i+1}$ induces a cycle of length four in $G$ and therefore $u'_{i}u'_{i+1}\in E(G)$. Now for any $i,j\in [k]$ with $|i-j|>1$, the vertices $u'_i$ and $u'_j$ are non-adjacent. Hence $Q$ induces a path in $G$.
\end{proof}

% Similarly, we have the following observation.

% \begin{observation}[\cite{caldam2020}]\label{obs:horizontal-neighbour-path}
% Let $P$ be a horizontal corner path of $G$ such that the $y$-coordinates of all vertices of $P$ is $\ell$. Consider the set $Q = \{ v\in V(G)\colon v\notin V (P ), N (v) \cap P \neq \emptyset\}$. Then $y$-coordinates of all vertices of $Q$ are the same and equals to $\ell-1$ or $\ell+1$. Moreover, $Q$ induces a path in $G$.
% \end{observation}

\begin{lemma}\label{lem:corner-path-props}
      Let $G$ be a solid grid graph and $\mathcal{R}$ be a grid embedding of $G$. Two distinct corner paths $P$ and $Q$ have the following properties.
    \begin{enumerate}[label=(\alph*)]
        \item\label{it:1} $E(P)\cap E(Q)=\emptyset$.
        \item\label{it:2} $|V(P)\cap V(Q)|\leq 1$.
        \item\label{it:3} If there exists a vertex $z\in V(P)\cap V(Q)$ then $z$ must be an end-vertex of both $P$ and $Q$.
    \end{enumerate}
\end{lemma}

\begin{proof}

     % Let $v_1$ and $v_2$ be the two end-vertices of $P$. Similarly, let $w_1$ and $w_2$ be the end-vertices of $Q$. 
     
     % Let $\mathcal{R}$ be a grid embedding of $G$. 
     Without loss of generality, assume that $P$ is a vertical corner path. 
     
     Now we shall prove \ref{it:1}. For the sake of contradiction, suppose $P$ and $Q$ share an edge $e$. Since $P$ is a vertical corner path, $Q$ must be a vertical corner path containing the edge $e$. This implies $Q$ must have the same end-vertices as $P$, and all intermediate vertices of $Q$ must lie in $P$ as well. Therefore $P=Q$, which is a contradiction. 

    Now we shall prove \ref{it:2}. Suppose $P$ and $Q$ share two vertices $z_1,z_2$, and without loss of generality, $y$-coordinate of $z_1$ is greater than that of $z_2$. Observe that $Q$ must be a vertical path containing the subpath of $P$ between $z_1,z_2$. This would imply that $P$ and $Q$ share an edge contradicting \ref{it:1}. 
    
    Now we shall prove \ref{it:3}. Suppose $P$ and $Q$ share exactly one vertex $z$. Then, observe that $Q$ must be a horizontal corner path. Now we have the following cases. 

    \begin{enumerate}
        \item  Suppose $z$ is neither an end-vertex of $P$ nor of $Q$. Then $z$ must have two neighbours, say $a,a'$, whose $y$-coordinate is the same as that of $z$. Similarly, $z$ must have two neighbours, distinct from both $a,a'$ whose $x$-coordinate is the same as that of $z$. Therefore, $z$ must have degree four in $G$, a contradiction to the fact that $z$ is a vertex of some corner path.

        \item Suppose $z$ is an end-vertex of one of the corner paths but not of the other. But this would directly contradict the definition of a corner vertex (recall that a corner vertex must have degree two in $G$ and other vertices of a corner path must have degree three in $G$).
        % have one neighbour, say $a$, whose $y$-coordinate is the same as that of $z$. Similarly, $z$ must have two neighbours, both distinct from $a$ whose $x$-coordinate is the same as that of $z$.  Therefore, $z$ must have degree three in $G$, a contradiction to the fact that $z$ is an end-vertex of a corner path ($P$). 
        
        % \item Due to symmetry, a similar argument as above proves the impossibility of $z$ being an end-vertex of $Q$ but not of $P$.
  
  \end{enumerate}
    
    Hence, $z$ is an end-vertex of both $P$ and $Q$.
\end{proof}

\begin{definition}
  Let $G$ be a solid grid graph and $\mathcal{R}$ be a grid embedding of $G$. For an arbitrary maximal corner sequence $S=(u_1,u_2,\ldots,u_{|S|})$ and $1\leq i\leq |S|-1$, let $T_{i}$ denote the corner path between $u_i$ and $u_{i+1}$ and let $\Phi(S)=\displaystyle\bigcup\limits_{i=1}^{|S|-1} V(T_i)$.
\end{definition}

\begin{lemma}\label{lem:disjoint}
      Let $G$ be a solid grid graph and $\mathcal{R}$ be a grid embedding of $G$. Let $S$ and $S'$ be two {distinct} maximal corner sequences. Then $\Phi(S) \cap \Phi(S') = \emptyset$. 
\end{lemma}
\begin{proof}
    Let $S=(u_1,u_2,\ldots,u_{|S|})$ and $1\leq i\leq |S|-1$, let $T_{i}$ denote the corner path between $u_i$ and $u_{i+1}$. Similarly,  let $S'=(u'_1,u'_2,\ldots,u'_{|S'|})$ and $1\leq i\leq |S'|-1$, let $T'_{i}$ denote the corner path between $u'_i$ and $u'_{i+1}$. Now suppose $\Phi(S) \cap \Phi(S') \neq \emptyset$. Then there exist $i$ and $j$ such that the two corner paths $T_i$ and $T'_j$ has a vertex $z$ in common. Due to Lemma~\ref{lem:corner-path-props}\ref{it:3}, we know that $z$ is an end-vertex of both $T_i$ and $T'_j$. Since $z$ has degree $2$ in $G$, we conclude that $z\in \{u_1,u_{|S|}\}$ and $z\in \{u'_1,u'_{|S|}\}$. But this contradicts the maximality of $S$ and $S'$.
\end{proof}

\subsubsection{Lower bound}\label{subsec:opt}

 Chakraborty et al.~\cite{caldam2020} proved the following observations and lemma using a slightly different terminologies. For completion we provide the proofs.

% \begin{observation}[\cite{caldam2020}]\label{obs:neighbour-path}
% Let $P$ be a corner path of $G$. Consider the set $Q = \{ v\in V(G)\colon v\notin V (P ), N (v) \cap P \neq \emptyset\}$. Then $Q$ induces an isometric path in $G$. 
% \end{observation}
% \begin{proof}
%     Without loss of generality assume $P=u_0 u_1 \ldots u_t$ be a vertical corner path. Then $Q=u'_0 u'_1 \ldots u'_t$ such that $N(u_i)\setminus V(P)=\{u'_i\}$. Therefore, all vertices of $Q$ must have the same $x$-coordinate. Since $P$ is a connected path, $Q$ induces an isometric path in $G$. 
% \end{proof}

\begin{lemma}[\cite{caldam2020}]\label{lem:cardinality}
  Let $G$ be a solid grid graph and $\mathcal{R}$ be a grid embedding of $G$. Any geodetic set of $G$ contains at least one vertex from each corner path.
\end{lemma}

\begin{proof}
Let there be a corner path $P$ and a geodetic set $X$ of $G$ such that $V(P)\cap X=\emptyset$. Without loss of generality, we assume $P$ to be a vertical corner path. Let $u$ be the end-vertex of $P$ with the larger $y$-coordinate. Now, consider two vertices $a, b \in X$ such that there is an isometric path $P'$ between $a$ and $b$ that contains $u$. Observe that $P'$ can be expressed as $P' = a~c_1 ~c_2\ldots c_t ~d ~f_1 ~f_2\ldots f_{t'} ~u ~g ~h_1~ h_2 \ldots h_{t''} ~b$ such that $\{f_1, f_2, \ldots , f_{t'}\} \subseteq V (P)$, $d\in N(f_1)\setminus V(P)$, $g\in N(u)\setminus V(P)$.

Now consider the set $Q=\{v\in V(G)\colon v \notin V (P ), N (v) \cap P \neq \emptyset\}$. Observe that $\{d, g\} \subseteq Q$ and therefore, due to Lemma~\ref{lem:vertical-neighbour-path}\ref{it:b}, the $x$-coordinate of both $d$ and $g$ is the same. Let the $y$-coordinate of $d$ and $g$ be $\ell$ and $\ell'$, respectively. Due to Lemma~\ref{lem:vertical-neighbour-path}\ref{it:c} and Observation~\ref{obs:iso-path}, we have that there is a isometric path between $d$ and $g$ of length $|\ell - \ell'|$. But the subpath of $P'$ between $d$ and $g$ is of length $|\ell- \ell'|+2$, which is not isometric, and, therefore, cannot be a subpath of an isometric path. Hence, this contradicts the fact that $P'$ is an isometric path.
\end{proof}

We have the following {lemma}.

\begin{lemma}\label{lem:cardinlaity-2}
  Let $G$ be a solid grid graph and $\mathcal{R}$ be a grid embedding of $G$. Let $\mathcal{S}$ be the set of all {distinct} maximal corner sequences of $G$, and let $t$ be the number of vertices of $G$ with degree $1$. Then, $gn(G)\geq t+\sum_{S\in \mathcal{S}} \left\lfloor \frac{|S|}{2}\right\rfloor$.
\end{lemma}

\begin{proof}
Let $X$ be a minimum geodetic set of $G$ and $V_1$ be the set of all vertices with degree one. Observe that $V_1\subseteq X$ and therefore $gn(G)\geq t$. Recall that for an arbitrary maximal corner sequence $S=u_1,u_2,\ldots,u_{|S|}\in \mathcal{S}$ and for $1\leq i\leq |S|-1$, $T_{i}$ denotes the corner path between $u_i$ and $u_{i+1}$ and $\Phi(S)=\displaystyle\bigcup\limits_{i=1}^{|S|-1} V(T_i)$. Lemma~\ref{lem:cardinality} implies that for each $1\leq i<|S|$, at least one vertex of the corner path $T_i$ must belong to $X$. Lemma~\ref{lem:corner-path-props}\ref{it:3} implies that no vertex of $G$ lies in more than two corner paths. Therefore, $X$ must contain at least $\left\lfloor \frac{|S|}{2} \right\rfloor$ vertices from $\Phi(S)$. 

Now, let $S_1, S_2, \ldots S_t$ be the maximal corner sequences. Observe that for any $\{i,j\} \subseteq \{1,2,\ldots, |\mathcal{S}|\}$, due to Lemma~\ref{lem:disjoint}, $\Phi(S_i)\cap \Phi(S_j)=\emptyset$. Now, previous arguments imply that $X$ must contain a set $V_2$ whose cardinality is at least $\sum_{S\in \mathcal{S}} \left\lfloor \frac{|S|}{2}\right\rfloor$. Moreover, $V_2\subseteq \displaystyle\bigcup\limits_{i=1}^t \Phi(S_i)$. Since no vertex of $V_1$ is a vertex of any corner path, $V_1\cap V_2=\emptyset$. Therefore, $gn(G)=|X|\geq t+\sum_{S\in \mathcal{S}} \left\lfloor \frac{|S|}{2}\right\rfloor$.
\end{proof}

\subsection{Proof of Theorem~\ref{thm:solid-grid}} \label{subsec:solid-conclude}

Now, proof of Theorem~\ref{thm:solid-grid} follows from Observation~\ref{obs:caardinality-solid}, Lemma~\ref{lem:feasibility} and Lemma~\ref{lem:cardinlaity-2}.
{Moreover, we have the following structural lemma, as a consequence of Observation~\ref{obs:caardinality-solid}, Lemma~\ref{lem:feasibility}, and Lemma~\ref{lem:cardinlaity-2}, that provides a tight bound on the geodetic number of a solid grid in terms of the cardinality of its corner sequences. 
\begin{lemma}\label{lem:tightBound}
   Let $G$ be a solid grid graph and $\mathcal{R}$ be a grid embedding of $G$. Let $\mathcal{S}$ be the set of all maximal corner sequences of $G$, and let $t$ be the number of vertices of $G$ with degree $1$. Then, $gn(G) = t+\sum_{S\in \mathcal{S}} \left\lfloor \frac{|S|}{2}\right\rfloor$.
\end{lemma}

The above (structural) lemma could be of independent interest. 

\medskip\noindent\textbf{Proof of Corollary~\ref{C:algo}} The proof follows from Observation~\ref{obs:caardinality-solid}, Lemma~\ref{lem:edgeFeasibe}, and Lemma~\ref{lem:cardinlaity-2}.

}

\section{Conclusion}

We proved that \textsc{Minimum Geodetic Set} is NP-hard on partial grids and that \textsc{Minimum Geodetic Set} admits linear time algorithm on solid grids. An interesting question is whether \textsc{Minimum Geodetic Set} is FPT on partial grids (or planar graphs) parameterized by the solution size? Another interesting question is the existence of constant factor approximation algorithms for \textsc{Minimum Geodetic Set} on partial grids or planar graphs. 

Finally, we think that studying the approximability and parameterized (in)tractability of related problems like \textsc{Isometric Path Cover}~\cite{chakraborty2022}, \textsc{Strong Geodetic Set}~\cite{lima2022computational}, and \textsc{Geodetic Hull}~\cite{hernando2005steiner} on planar graphs is another interesting direction of research.

\medskip \noindent \textbf{Acknowledgement:} We thank the reviewers for reading the paper carefully and for all their suggestions that improved the paper substantially. We thank one of the reviewers for providing the counter example shown in Figure~\ref{fig:converse-not-true}. 

\bibliographystyle{alpha}
\bibliography{references}

\end{document}